\documentclass[10pt, journal]{IEEEtran}
\IEEEoverridecommandlockouts 

\usepackage{amsthm, amsmath,amssymb,amsfonts}
\usepackage{pifont}                                       %
\usepackage{caption}  
\usepackage{graphics}
\usepackage{graphicx}
\usepackage{gensymb}
\usepackage[noend,  ruled, linesnumbered]{algorithm2e}
\usepackage{color,xcolor}
\usepackage{cases} 
\usepackage{url}
\usepackage{float}  
\usepackage{booktabs}
\usepackage[caption=false,font=footnotesize,labelfont=sf,textfont=sf,subrefformat=parens]{subfig}
\usepackage{enumitem}                                      
\usepackage{textcomp}                                      
\usepackage{stfloats} 
\usepackage{multirow}
\usepackage{makecell} 
\usepackage{cite}
\usepackage{orcidlink}
\usepackage{siunitx} 

\usepackage{comment}                                       
\usepackage{soul}                                          
\soulregister\cite7
\soulregister\ref7
\soulregister\pageref7
\usepackage{etoolbox}                                      
\usepackage{hyperref}
\hypersetup{
    colorlinks = true,
    citecolor  = blue,
    linkcolor  = blue,
    urlcolor   = blue,
}
\usepackage{cleveref}
\Crefformat{figure}{Fig.~#2#1#3}                           
\Crefname{subfigure}{Fig.}{Figs.}
\Crefname{figure}{Fig.}{Figs.}
\Crefformat{table}{TABLE~#2#1#3}                           

\newtheorem{theorem}{Theorem}

\newtheorem{example}{Example}
\newtheorem{property}{Property}

\newcommand{\kw}[1]{{\ensuremath {\mathsf{#1}}}\xspace}
\newcommand{\marked}[1]{\textcolor{red}{#1}}

\newcommand{\secd}[1]{\SI{#1}{\second}}
\newcommand{\cmark}{\ding{51}}%
\newcommand{\xmark}{\ding{55}}%

\newcommand{\ie}{\emph{i.e.,}\xspace}
\newcommand{\eg}{\emph{e.g.,}\xspace}
\newcommand{\wrt}{\emph{w.r.t.}\xspace}

\newcommand{\rw}{\emph{rewrite}\xspace}
\newcommand{\rf}{\emph{refactor}\xspace}
\newcommand{\rs}{\emph{resubstitution}\xspace}
\newcommand{\rss}{\emph{resub}\xspace}
\newcommand{\rwz}{\emph{rewrite -z}\xspace}
\newcommand{\rfz}{\emph{refactor -z}\xspace}

\newcommand{\aff}{\kw{AFF}}
\newcommand{\inc}{\kw{incLM}}
\newcommand{\incpq}{\kw{pqLM}}
\newcommand{\incour}{\kw{boundLM}}
\newcommand{\dynto}{\kw{dynTO}}
\newcommand{\dynlev}{\kw{dynLev}}
\newcommand{\dynrl}{\kw{dynRL}}
\newcommand{\find}{\kw{findInv}}
\newcommand{\reorder}{\kw{reordInv}}
\newcommand{\insertafter}{\kw{insertAfter}}

\newcommand{\minisection}[1]{\vspace{.06in}\noindent{\textbf{#1}}.}

\begin{document}

\twocolumn
\title{
Bounded Dynamic Level Maintenance for Efficient Logic Optimization
}

\author{
  Junfeng Liu,~\ 
  Qinghua Zhao,~\ 
  Liwei Ni,~\ 
  Jingren Wang,~\
  Biwei Xie, ~\ 
  Xingquan Li,~\ 
  Bei Yu,~\
  Shuai Ma
  
  \thanks{This paper is currently under review by IEEE TRANSACTIONS ON COMPUTERS.}
  \thanks{Junfeng~Liu, Liwei~Ni and Xingquan~Li are with the Department of Optoelectronic Information and Optical Fiber Communication, Pengcheng Laboratory, Shenzhen, China.}
  \thanks{Qinghua~Zhao are with the School of Artificial Intelligence and Big Data, Hefei University, Hefei, China.}
  \thanks{Jingren~Wang is with the Microelectronics Thrust,  Hong Kong University of Science and Technology (Guangzhou), Guangzhou, China.} 
  \thanks{Biwei Xie is with the Institute of Computing Technology Chinese Academy of Sciences, Beijing, China}
  \thanks{Bei~Yu is with the Department of Computer Science and Engineering, The Chinese University of Hong Kong, Hong Kong SAR.}
  \thanks{Shuai~Ma is with SKLCCSE Lab, Beihang University, Beijing, China.}
}


\maketitle


\begin{abstract}
Logic optimization constitutes a critical phase within the Electronic Design Automation (EDA) flow, 
essential for achieving desired circuit power, performance, and area (PPA) targets. 
These logic circuits are typically represented as Directed Acyclic Graphs (DAGs),
where the structural depth, quantified by node level, critically correlates with timing performance. 
Modern optimization strategies frequently employ iterative, local transformation heuristics (\emph{e.g.,} \emph{rewrite}, \emph{refactor}) directly on this DAG structure.
As optimization continuously modifies the graph locally, node levels require frequent dynamic updates to guide subsequent decisions.
However, a significant gap exists: existing algorithms for incrementally updating node levels are unbounded to small changes. 
This leads to a total of worst complexity in $O(|V|^2)$ for given local subgraphs $\{\Delta G_i\}_{i=1}^{|V|}$ updates on DAG $G(V,E)$. 
This unbounded nature poses a severe efficiency bottleneck, hindering the scalability of optimization flows, particularly when applied to large circuit designs prevalent today.
In this paper, we analyze the dynamic level maintenance problem endemic to iterative logic optimization, framing it through the lens of partial topological order. 
Building upon the analysis, 
we present the first bounded algorithm for maintaining level constraints,  with $O(|V| \Delta \log \Delta)$ time for a sequence $|V|$ of updates $\{\Delta G_i\}$, where $\Delta = \max_i \|\Delta G_i\|$ denotes the maximum extended size of $\Delta G_i$.
Experiments on comprehensive benchmarks show our algorithm enables an average 6.4$\times$ overall speedup relative to \rw and \rf, driven by a 1074.8$\times$ speedup in the level maintenance, all without any quality sacrifice. 
\end{abstract}

\begin{IEEEkeywords}
logic optimization, dynamic level maintenance, incremental graph computation.
\end{IEEEkeywords}  
\IEEEpeerreviewmaketitle
 
\section{Introduction}
\label{sec:intro}

\IEEEPARstart{L}ogic synthesis remains a fundamental cornerstone in electronic design automation (EDA), serving as the critical bridge between high-level hardware description languages and physical implementation~\cite{abc, sis, heinz2019dag, li2024ieda, liu2023aimap, liu2024itmap, dara13computing, wang2025fgnn}.
Within synthesis, multi-level logic optimization critically determines circuit performance, power consumption, and area (PPA), \eg shortening critical paths to improve timing and reducing gate counts to minimize area.
As circuits grow in complexity, the efficiency of these optimization techniques has become essential for meeting design requirements within practical development timeframes~\cite{li2023effsyn, amaru2017logic}.

Due to their effective balance of expressiveness and simplicity~\cite{abc, bryant86graph, luca2016maj}, directed acyclic graphs (DAGs) have become the dominant representation in contemporary logic synthesis frameworks, \eg And-Inverter Graphs (AIGs) and Majority-Inverter Graphs (MIGs). 
Logic optimization on DAGs typically focuses on two primary objectives: reducing circuit graph size (\eg node count) and constraining or reducing graph level,  which directly correspond to the improved area and delay in circuit implementations~\cite{tempia2023improving}.

\begin{figure}[tb!] 
 \subfloat[Typical local transformation-based synthesis flow. \label{subfig:synflow}]{
		\includegraphics[width=0.44\linewidth]{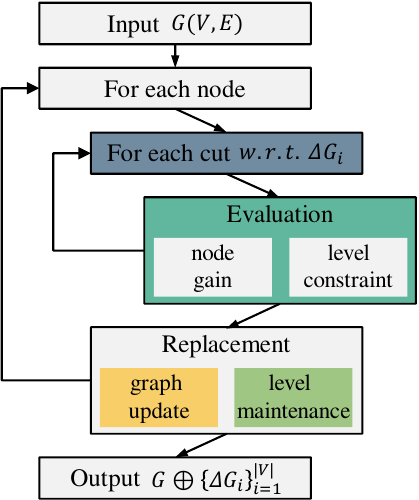}}
  \quad
  \subfloat[Runtime breakdown of local operator \rw~\cite{alan2006dag}. \label{subfig:time_proportion}]{
		\includegraphics[width=0.52\linewidth]{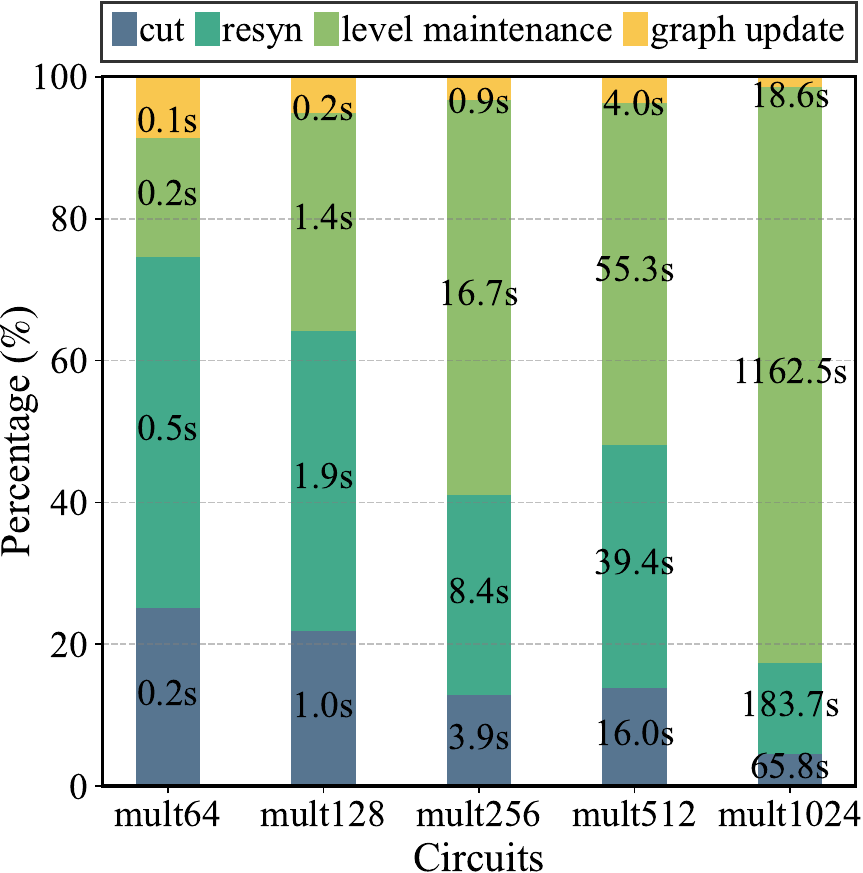}}
\caption{Motivation example: synthesis flow overview and its runtime breakdown.}  
\label{fig:motivation}    
\end{figure}  

However, optimizing the logic circuit \wrt minimizing node counts or DAG levels, is inherently NP-hard, rendering exact solutions computationally intractable~\cite{villa2012synthesis}. Thus, existing synthesis flows typically employ iterative local transformations on DAG representation, to efficiently obtain near-optimal results~\cite{alan2006dag, heinz2019dag, yu2018fast}.
As shown in~\Cref{subfig:synflow}, for each node in the input circuit graph (\ie $G$), the algorithms replace subgraphs with functionally equivalent alternatives using Boolean or algebraic methods, \eg factoring of expressions~\cite{alan2006dag}, or SAT-based approaches~\cite{brayton2006scalable}, where each subgraph is generally delimited by node cuts. 
When a local optimality subgraph (\ie $\Delta G$) is identified, the original graph $G$ replaces the relevant portion with  $\Delta G$, 
and the levels of the resulting graph $G \oplus \Delta G$ are updated to maintain level constraints.
For instance, in the widely used synthesis tool ABC~\cite{abc}, local transformation operators such as \rw, and \rf are employed to identify and apply beneficial subgraphs, optimizing network size while adhering to level constraints~\cite{alan2006dag,  brayton2006scalable}.


The dynamic nature of logic optimization and the computational expense of level updates underscore the importance of \emph{bounded} dynamic level computation. Two key challenges drive this need:
First, modern combinational designs have grown extremely large due to complex logic functions, often with hundreds of millions of nodes~\cite{amaru2017logic, epfl}. 
During synthesis, these graphs undergo constant updates through local transformation operations.
Computing level for the entire updated graph $G \oplus \Delta G$ from scratch after each update $ \Delta G$ becomes prohibitively expensive as design sizes increase.
Second, even dynamic methods that update only affected nodes face significant challenges, as existing dynamic level maintenance methods are \emph{unbounded} \wrt local updates $\Delta G$~\cite{Ramalingam1996dynamic,fan2021inc}.
Specifically, even a small $\Delta G_i$ to $G$ can potentially affect the levels across almost the entire graph,
causing existing incremental methods running in $O(|V|^2)$ for $\{\Delta G_i\}_{i=1}^{|V|}$ total updates~\cite{abc,KatrielMH05}.
However, the resynthesis step among most local operators only runs in $O(|V||C|)$, \eg \rw, where $|V|$ and $|C|$ represent the number of nodes and cuts per node, respectively.
As demonstrated in~\Cref{subfig:time_proportion}, this unbounded level computation becomes increasingly problematic with large designs. 
The five multipliers of increasing bit widths contain 0.04, 0.2, 0.7, 2.9, and 11.7 million nodes, respectively.
The proportion of time spent on level updates grows dramatically from \secd{0.2} out of \secd{1.0}  total runtime to \secd{1162.5} out of \secd{1430.5}. 
This escalating cost highlights the critical need for more efficient level maintenance methods in modern logic synthesis flows.

Several studies have attempted to improve logic optimization efficiency through two approaches: artificial intelligence enhancement~\cite{li2023effsyn, bai2025graph, wang2024ls} and GPU-based acceleration~\cite{li2023recursion, lin2022novelrewrite, liu2024unified, possani2018unlocking}. 
In the first category, Li~\emph{et al.}~\cite{li2023effsyn} employ prediction models to prune candidate cuts, while Bai~\emph{et al.}~\cite{bai2025graph} and Wang~\emph{et al.}~\cite{wang2024ls} introduce classification tasks and learned symbolic functions to prune candidate nodes in resynthesis of \rw and \emph{mfs} operators, respectively. 
In the second category, research by~\cite{li2023recursion, lin2022novelrewrite, liu2024unified, possani2018unlocking} proposes various fine-grained unlocking strategies to exploit parallelism at node or cut level, accelerating \rw efficiency on GPUs or in parallel environments.
Despite these significant efforts, all these methods focus solely on node count reduction as their optimization objective. 
None addresses the more challenging problem of improving efficiency under level constraints, despite this being a more general target~\cite{amaru2017logic}.
Note that, while Katriel~\emph{et al.}~\cite{KatrielMH05} propose a general dynamic longest path maintenance algorithm, their approach has not yielded efficiency improvements in the level-constrained logic optimization.

To this end, we design an efficient dynamic level maintenance algorithm for logic optimization.
To our knowledge, this is the first work to establish a new paradigm that uses dynamic graph computation analysis to theoretically enhance synthesis efficiency. 
The main contributions are as follows. 
 
\begin{enumerate}   
\item We analyze the key insights for bounded level maintenance algorithms, by framing it through the lens of partial topological order.

\item By the analysis,  we present \incour,  which dynamically maintains partial topological order (\dynto), node levels (\dynlev), and reverse levels (\dynrl).
Our \incour is the first bounded algorithm running in $O(|V|\Delta\log \Delta)$  time for $|V|$ updates $\{\Delta G_i\}$, where $\Delta = \max_i \|\Delta G_i\|$ denotes the maximum extended size of $\Delta G_i$. 

\item On large-scale benchmarks (0.214$\times 10^6$ to 41.9$\times 10^6$ nodes),  \incour achieves a 6.4$\times$ average speedup for \rw and \rf operators, with a 1074.8$\times$ acceleration in level maintenance, with preserved result quality.
We verify the \incour's boundedness, scalability, and robustness across different configurations. 
\end{enumerate}


The remainder of this paper is organized as follows: \Cref{sec:pre} provides the necessary background and the problem definition. 
\Cref{sec:analysis} analyzes the key design insights of dynamic level computation.
\Cref{sec:method_continuous}  presents the overview of the dynamic bounded level maintenance for local transformation-based synthesis.
\Cref{sec:method_single}  discusses the details of the bounded algorithm.
\Cref{sec:exp} presents experimental results and in-depth analyses, followed by the conclusion in~\Cref{sec:conclusion}.

\section{Preliminary and Problem Definition}
\label{sec:pre}

In this section, we first introduce basic concepts, followed by the logic optimization flow and the dynamic level maintenance, and conclude with the problem definition.

\subsection{Key Terminology}
\minisection{Boolean Circuit} 
A Boolean circuit $G(V, E)$ is a directed acyclic graph (DAG), where each node corresponds to a logic gate and each directed edge $(x,y)$ represents a wire connecting node $x$ to node $y$~\cite{abc}. 
The fanin and fanout of a node $x \in V$ are its incoming and outgoing edges, respectively. 
The primary inputs (PIs) are nodes without incoming edges, primary outputs (POs) are nodes whose computed functions constitute the signals provided to the circuit's environment. 
The \emph{level/delay} of the circuit is the largest path to any POs.
The AND-Inverter Graph (AIG) serves as a prevalent circuit representation, utilizing only 2-input AND gates as nodes, where inverters are associated with the edges.

 
\minisection{Cut} 
A cut $C$ associated with a node $x$ in Boolean circuit $G$ is a set of nodes $\{c_1, \cdots, c_m\}$ such that every path from a PI to $x$ traverses at least one node in $C$. 
A $k$-feasible cut is defined as a cut $C$ whose size does not exceed a predefined integer $k$, \ie $|C| \leq k$, where $k$ is typically 4 or 6 in practice. 
This constraint limits the logic function's complexity \wrt the subgraph induced by the cut, facilitating rapid logic optimization through local transformations.

\minisection{(Partial) Topological Ordering}
A topological order on a DAG is a strict \emph{total order} relation ``$\prec$'' defined on the set $V$ such that for every edge $(x, y) \in E$, it holds that $x \prec y$. 
For dynamic synthesis scenarios where nodes are processed incrementally, we extend this concept to \emph{partial topological order}. 
Given a subset $V^{-} \subseteq V$, a partial topological order on $V^{-}$ is a strict total order ``$\preccurlyeq$" such that for all $(x, y) \in E$ with $x, y \in V^{-}$,  $x \preccurlyeq y$.
To represent the topological order, standard implementations assign integer labels $ord(x) \in \{1, \ldots, |V|\}$ to maintain this ordering,  but such static assignments are not flexible for dynamic scenarios.
Although sophisticated data structures \eg ordered lists achieve $O(1)$ amortized time complexity for dynamic order test and update operations, they incur significant implementation overhead and high constant factors~\cite{bender2002two, liu2025inc}.

\subsection{Local Transformation-based Logic Optimization} 
The typical local transformation-based synthesis flow is illustrated in~\Cref{subfig:synflow}, which serves as the basis for modern logic synthesis engines~\cite{abc}.
Building upon AIG representations, we briefly introduce the three main iterative steps.

\minisection{Cut Enumeration}
The process begins by identifying candidate regions for optimization. 
For a given node $x$ in the AIG, it computes the set of $k$-feasible cut using a bottom-up or top-bottom approach that combines the set of cuts from $x$'s fanins.
The enumerated cuts are typically pruned by a heuristic function to reduce candidate cuts for efficiency. 
Each cut $C$ induces a subgraph rooted at $x$ whose inputs are the nodes in $C$,  representing the local logic function feeding into $x$.

\minisection{Subgraph Evaluation}
This step evaluates potential optimizations for each subgraph. This involves applying various local transformation operators, \eg \rw, \rf. 
For a detailed review, please refer to literature of~\cite{micheli1994synthesis, alan2006dag, heinz2019dag,brayton2006scalable}.
Specifically, in \rw, 
it attempts to replace the subgraph induced by the cut with a structurally different but functionally equivalent subgraph. 
These equivalent subgraphs are often looked up from a precomputed library, which offers a canonicalized logic function of the cut using NPN equivalence (negation of outputs, permutation and negation of inputs)~\cite{alan2006dag}.
In \rf, it evaluates the candidate logic structures of the cone by a factored form of the root function, with deeper and less structurally biased adjustments compared to \rw~\cite{amaru2018improve}.

During evaluation, each potential subgraph generated by these operators is assessed based on the reduction of AIG nodes while satisfying level constraints.
A common level constraint check is formulated as:
\begin{equation}
\label{eq:level_const}
 \mathcal{L}(x^{\prime}) \leq \mathcal{L}_{\max} - \mathcal{R}(x)  
\end{equation}
where $\mathcal{L}(x^{\prime})$ denotes the level of replacement node $x^{\prime}$ (logic equivalent to $x$),
$\mathcal{L}_{\max}$ represents the circuit's maximum allowed level, and $\mathcal{R}(x)$ indicates the reverse level of $x$ measured from the POs.
This constraint ensures that the update does not unduly increase the circuit's critical path delay.

\minisection{Subgraph Replacement}
If the evaluation identifies a transformation offers acceptable node gain and adheres to level constraints (cf.~\Cref{eq:level_const}), this step implements the changes within the AIG.
The original subgraph induced by the cut (excluding the cut nodes themselves, and potentially reusing the root $x$ if the transformation preserves it) is replaced by the new optimized subgraph ($\Delta G$).
After graph update, the level $\mathcal{L}(\cdot)$ and reverse level $\mathcal{R}(\cdot)$ of $G \oplus \Delta G$ should be recomputed for all affected nodes, 
for subsequent constraint evaluations.

\begin{figure}[tb!]
    \centering
    \includegraphics[width=.86\columnwidth]{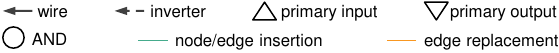}
    \newline 
    \subfloat[Example AIG. \label{subfig:aig}]{
		\includegraphics[width=0.32\columnwidth]{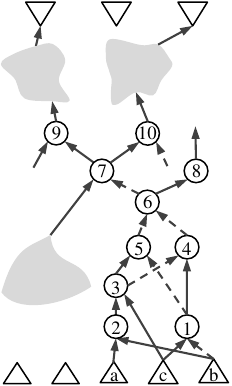}}
    \quad\quad
    \subfloat[Optimized AIG.  \label{subfig:aig_rw}]{
		\includegraphics[width=0.32\columnwidth]{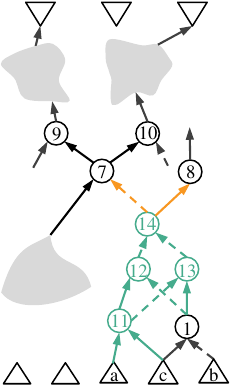}} 
    \caption{ Example of a local optimization on AIG.}
    \label{fig:example_lo}
\end{figure}

We next illustrate the flow with an example.
 \begin{example} 
As shown in~\Cref{fig:example_lo}, consider node $6$ with a $3$-feasible cut $\{a, b, c\}$ that implements function $f = abc + \overline{b}c$.
Through NPN equivalence matching (employed in \rw~\cite{alan2006dag}) or factorization techniques (employed in \rf~\cite{zonglin2024parallel}), a logically equivalent subgraph is identified.
During the evaluation, a candidate replacement subgraph rooted at node $14$ implementing $f = ac + \overline{b}c$ is selected.
This transformation reduces both node count and level by $1$, as node $6$ has level $4$ while node $14$ has level $3$.
The level and reverse level for affected nodes are required to be recomputed. 
\end{example}

\subsection{Dynamic Level Maintenance}
The dynamic computation of level and reverse level becomes evident due to the iterative subgraph replacement.
During the replacement, the substitution of a node with a logic-equivalent subgraph at a different level necessitates incremental level updates for affected nodes in its transitive fanout, \eg node $14$ and its transitive fanout in~\Cref{subfig:aig_rw}.  
  
Algorithm~\inc presents a basic incremental level computation method updating an AIG's level map $\mathcal{L}(\cdot)$ after a subgraph modification $\Delta G$~\cite{abc}, as shown in~\Cref{alg:level_basic}.
It employs a level-based traversal, starting from the nodes affected by $\Delta G$, with a level-indexed data structure $levelVec$ (lines 1-2).
Nodes are processed level-by-level to recompute and update the level (lines 3-5, 12).
When a node's fanout $f$ requires a level update, the level map $\mathcal{L}(f)$ is updated, and $f$ is added to the level-indexed data structure $levelVec$ corresponding to its new level for subsequent processing, propagating changes incrementally (lines 6-11).
  
The time complexity of \inc is $O( \max(\mathcal{L}_{\max}, \|\aff\|)$ for dynamically updating the starting nodes of $\Delta G$, 
where $\mathcal{L}_{\max}$ refers to the maximum level in $G$, and $\aff$ represents the set of nodes whose levels are affected by the subgraph replacement. 
Thus, \inc runs in $O(|V|^2)$ time for $\{\Delta G_i\}_{i=1}^{|V|}$ total updates.
Here, $\|\aff\|$ denotes the extended size of $\aff$, comprising the affected nodes along with their immediate fanins and fanouts.
Throughout this paper, the notation $\|\cdot\|$ consistently denotes the extended size, which captures \emph{the size of the changes in the input and output}, following the convention in incremental graph algorithms~\cite{KatrielMH05,Ramalingam1996dynamic,liu2025inc}.


However, even when the subgraph replacement does not result in actual updates, algorithm~\inc still run in $O(\mathcal{L}_{\max})$, \ie it cannot be bounded solely by the affected nodes.
To address this limitation, Katriel \emph{et al.}~\cite{KatrielMH05} propose algorithm~\incpq using a priority queue to maintain nodes requiring level updates (rather than using level-indexed structures), while following the same update propagation as~\inc.
Algorithm~\incpq approach achieves a time complexity of $O(\|\aff\| + |\aff|\log{|\aff|})$~\cite{KatrielMH05}.

\begin{algorithm}[tb!]
\small
    \caption{\small A preliminary version of incremental level computation (\inc~\cite{abc})} \label{alg:level_basic} 
    \KwIn{AIG $G$, subgraph $\Delta G$, level map $\mathcal{L}$ for $G$  } 
    \KwOut{Updated level map $\mathcal{L}$ for $G \oplus \Delta G$}  
    $levelVec \gets$ Initialize level-indexed node vector\;
    Insert starting nodes from $\Delta G$ into  levels in $levelVec$\;
    \For{$i \gets 0$ \KwTo $|levelVec|-1$}{
        $curLevel \gets levelVec[i]$\;
        \lIf{$curLevel = \emptyset$}{\textbf{continue}}
        
        \ForEach{\textnormal{node} $x \in curLevel$}{
            \ForEach{\textnormal{node} $f \in \textnormal{fanout}(x)$}{
                $l_{\textnormal{new}} \gets 1 + \max\{\mathcal{L}(y) \mid y \in \text{fanin}(f)\}$\;
                
                \If{$l_{\textnormal{new}} \neq \mathcal{L}(f)$}{
                    $\mathcal{L}(f) \gets l_{\textnormal{new}}$\;
                    Insert $f$ into $levelVec[l_{\textnormal{new}}]$\;
                }
            }
        }
    }
    \Return{$\mathcal{L}$ \textnormal{for nodes in} $G \oplus \Delta G$\;}
\end{algorithm}

\subsection{Problem Definition}
 
The dynamic level maintenance for logic optimization problem is defined as follows.

\textbf{Problem.} {\em
Consider a Boolean circuit represented as an AIG $G(V, E)$, and a sequence of logical-equivalent subgraph transformations $\{\Delta G_1,\Delta G_2, \ldots, \Delta G_{|V|}\}$.
Each transformation $\Delta G_i$ modifies the graph $G_{i-1}$ to produce $G_i = G_{i-1} \oplus \Delta G_i$. 
The level-constrained logic optimization problem is to efficiently update node levels $\mathcal{L}$ and reverse levels $\mathcal{R}$ such that 
 the maximum level of updated graph $G_{|V|}$ does not exceed that of $G$, while minimizing the cumulative computational cost.
}  

Note that, the number of subgraph transformations is bounded by the node count $|V|$ of the original graph $G$, 
as logic optimization attempts potential equivalent replacements at each node in $G(V,E)$. 
Unless otherwise specified, $G_0$ denotes the original graph $G$ in this paper.

Different from existing artificial intelligence enhancements~\cite{li2023effsyn, bai2025graph, wang2024ls} and GPU-based acceleration strategies~\cite{li2023recursion, lin2022novelrewrite, liu2024unified, possani2018unlocking}, 
we first tackle the fundamental $O(|V|^2)$ efficiency bottleneck of dynamic level maintenance from a theoretical perspective of complexity optimization.

The main notations are summarized in~\Cref{tab:notation}.


\section{Analysis of Dynamic Level Maintenance}
\label{sec:analysis}
As discussed, adhering to level constraints, particularly preventing increases in the logic network's level, is a fundamental requirement during logic optimization processes. 
This task can be formalized as the dynamic maintenance of node levels and reverse levels within an AIG during structural modifications.
This section analyzes the core insights in designing efficient dynamic level maintenance,  first examining affected regions from subgraph replacements, then identifying properties to reduce these regions for bounded algorithm designs.
%

\subsection{Analysis of Affected Region} 
\label{sec:analysis_affect}

To characterize the level and reverse level affected regions resulting from dynamic graph modifications during logic optimization, we analyze the impact of the update operations. While complex transformations, such as replacing a subgraph rooted at $x$ with a new, logically equivalent subgraph $\Delta G$ rooted at $x^{\prime}$, involve multiple changes, their effects can be decomposed into combinations of unit updates.
We consider \emph{w.l.o.g.} the following unit updates within $G \oplus \Delta G$.

\begin{table}[tb!]
\centering
\caption{Main Notation}
\label{tab:notation}
\setlength{\tabcolsep}{4pt} 
\resizebox{\columnwidth}{!} {
    \begin{tabular}{l|l} \toprule
    Notation & Description  \\  \midrule
    $G(V,E)$                      & Boolean circuit represented as an AIG           \\
    $\Delta G$                    & updates to $G$ (edge insertions, replacements, deletions)  \\ 
    $G \oplus \Delta G$           & circuit obtained by updating $\Delta G$ to $G$ \\ 
    $\Delta $                     & maximum subgraph size, $\max_{1 \leq i \leq |V|} |\Delta G_i|$ \\
    $\mathcal{L}/\mathcal{R}$     & level/reverse level maps for node set $V$  \\ 
    $\mathcal{I}_{x,x^{\prime}}$  & \makecell[l]{node set directly affected by edge insertions during \\ replacing nodes $x$ by $x^{\prime}$ } \\ 
    $\mathcal{P}_{x,x^{\prime}}$  & \makecell[l]{node set directly affected by edge replacements during \\replacing nodes $x$ by $x^{\prime}$ }  \\ 
    $\mathcal{D}_{x,x^{\prime}}$  & \makecell[l]{node set directly affected by edge deletions during\\ replacing nodes $x$ by $x^{\prime}$ }   \\  
    \bottomrule
    \end{tabular}
}
\end{table}

\begin{itemize}[label=\textopenbullet]  
  \item edge insertion: (\kw{insert} $e$ ),  adding an edge (possibly with a new node) constituting the new subgraph $\Delta G$ into $G$.
  \item edge deletion: (\kw{delete} $e$),  removing an edge (possibly with an existing node) from $G$, typically from the original subgraph rooted at $x$. 
  \item edge replacement: (\kw{replace} $e_1, e_2$), while implemented using unit \kw{delete}  $(x,\cdot)$  and \kw{insert} $(x^{\prime}, \cdot)$ on edges,
   it \emph{conceptually} uses redirection connections between logically equivalent nodes by \kw{replace} edges $(x,\cdot), (x^{\prime}, \cdot)$.
\end{itemize}
Any complete subgraph replacement operation can thus be achieved through a coordinated sequence of these unit updates.
Specifically, for a subgraph replacement from $x$ to $x^{\prime}$ with $\Delta G$, the process proceeds as follows:
The new subgraph $\Delta G$ is first incorporated into $G$ via \kw{insert} $e$.
The fanout edges of node $x$ are then redirected to $x^{\prime}$, \ie~\kw{replace} $(x,\cdot)$ $( x^{\prime}, \cdot)$.
Finally, the obsolete subgraph is removed by recursively traversing backward from node $x$ by \kw{delete} $e$, removing nodes and edges that lose all their fanouts.


 
After updating $G_{i-1}(V_{i-1}, E_{i-1})$ to $G_i(V_i, E_i)$ by replacing the subgraph at node $x$ with the new subgraph $\Delta G_i$ rooted at $x'$, we define three sets: $\mathcal{I}_{x,x'}$, $\mathcal{D}_{x,x'}$, and $\mathcal{P}_{x,x'}$. These sets identify nodes directly affected by \kw{insert}, \kw{delete}, and \kw{replace} operations, respectively, which serve as starting nodes for incremental level update algorithms.

$\mathcal{I}_{x,x^{\prime}}$ is the set of nodes directly affected by \kw{insert} edges. 
$\mathcal{I}_{x,x^{\prime}} = \{ n \mid n \in V_i  \wedge  n \notin V_{i-1} \}$, 
it comprises the nodes that are newly connected to $G_i$.
 
$\mathcal{D}_{x,x^{\prime}}$ is the set of nodes directly affected by \kw{delete}.
$\mathcal{D}_{x,x'} = \{n \mid n \in V_i \wedge \exists m, (m,n) \in E_{i-1} \wedge (m,n) \notin E_i \}$,    
it denotes the set of nodes that remain in $G_i$ but have lost at least one fanin during the deletion process.
 
$\mathcal{P}_{x,x^{\prime}}$ is the set of nodes directly affected by \kw{replace} edge.
$\mathcal{P}_{x,x'} = \{n \mid (x,n) \in E_{i-1} \wedge (x',n) \in E_i\}$, 
it comprises nodes whose fanins are redirected from node $x$ to node $x'$ during the subgraph replacement.
  
We observe that the nodes in $\mathcal{I}_{x,x^{\prime}}$ impact both level and reverse level computations, which represents a trivial case. The level of nodes in $\mathcal{I}_{x,x^{\prime}}$ can be easily computed using the definition $\mathcal{L}(n) = 1 + \max\{\mathcal{L}(f) \mid f \in \textnormal{fanin}(n) \}$ for each $n \in \mathcal{I}_{x,x^{\prime}}$, since these nodes are inserted following topological order. 
Besides,  the nodes in $\mathcal{I}_{x,x^{\prime}}$ are reachable from $x^{\prime}$. 
Thus, when we process the reverse level computation for $x^{\prime}$, the nodes in $\mathcal{I}_{x,x^{\prime}}$ are naturally incorporated into the update.

Thus, we focus on non-trivial cases for starting nodes of level and reverse level computation with \kw{replace} and \kw{delete}:
\begin{itemize}[label=\textopenbullet]  

  \item Compute $\mathcal{L}$: identify starting nodes $n \in \mathcal{P}_{x,x'}$  with inconsistent levels, \ie $\mathcal{L}(n) \neq 1 + \max\{\mathcal{L}(f) \mid f \in \textnormal{fanin}(n) \}$.

  \item Compute $\mathcal{R}$: identify starting nodes $n \in \{x'\} \cup \mathcal{D}_{x,x'}$ with inconsistent reverse levels, \ie $\mathcal{R}(n) \neq 1 + \max\{\mathcal{R}(f) \mid f \in \textnormal{fanout}(n)\}$.
  

\end{itemize}

Based on the identified starting nodes, we can characterize the entire affected regions requiring recomputation.

The affected region $\mathcal{A}_{\mathcal{L}}$ comprises all nodes visited during this forward propagation whose levels are potentially incorrect. 
\begin{equation}
\begin{aligned}
\mathcal{A}_{\mathcal{L}} = \Big\{ m \in V_i \mid{} & \exists n \in \mathcal{P}_{x,x'},  n \rightsquigarrow m, 
  \mathcal{L}_{\text{old}}(m)  \\ & \neq 1 + \max_{p \in \text{fanin}(m)} \mathcal{L}_{\text{new}}(p) \Big\} 
\end{aligned}
\label{eq:aff_region}
\end{equation}
where $n \rightsquigarrow m$ indicates reachability from $n$ to $m$, and $\mathcal{L}_{\text{old}}$ and $\mathcal{L}_{\text{new}}$ represent levels before and after the propagation. 
Recomputation is needed within this region until levels stabilize.

Similarly, incremental reverse level computation starts from the nodes whose fanouts are charged. 
The affected region $\mathcal{A}_{\mathcal{R}}$ comprises nodes visited during the backward propagation from these sources until reverse levels stabilize:
\begin{equation}
\begin{aligned}
\mathcal{A}_{\mathcal{R}} = \Big\{ m \in V_i \mid{} & \exists n \in \{x^{\prime}\} \cup \mathcal{D}_{x,x'},  m \rightsquigarrow n, 
  \mathcal{R}_{\text{old}}(m)  \\ & \neq  1 + \max_{p \in \text{fanout}(m)} \mathcal{R}_{\text{new}}(p)   \Big\}
\end{aligned}
\label{eq:aff_region_R}
\end{equation}
where $m \rightsquigarrow n$ denotes $n$ is reachable from $m$, $\mathcal{R}_{\text{old}}$ and $\mathcal{R}_{\text{new}}$ are reverse levels before and after the propagation.

We next illustrate these concepts with an example.

\begin{example}
As shown in~\Cref{fig:example_lo}, also consider the replacement of node $6$ by its logical equivalent node $14$.
\kw{Insert}: the subgraph $\Delta G$ implementing the logic of node $14$ is inserted into the graph $G$, \ie $\mathcal{I}_{6,14} = \{11, 12, 13, 14\}$, in green.
\kw{Delete}: the original node $6$ and its fanouts are deleted, and any nodes that become fanout-free are also recursively deleted, \ie $\mathcal{D}_{6,14} = \{1, c, a, b\}$.
\kw{Replace}: the composite \kw{delete} and \kw{insert} edge on the two logical equivalent nodes refer to replacement, \ie $\mathcal{P}_{6,14} = \{7, 8\}$, in orange.

The affected region for level computation $\mathcal{A}_{\mathcal{L}}$ requires recomputation starting from $\mathcal{P}_{6,14} = \{7, 8\}$ once their level are updated. 
$\mathcal{A}_{\mathcal{R}}$ requires recomputation starting from $\{14\} \cup \mathcal{D}_{6,14} = \{14, 1, c, a, b\}$ once their reverse level are modified.
Benefiting from the updated level $\mathcal{L}$ and reverse level $\mathcal{R}$, the level constraint in~\Cref{eq:level_const} is correctly checked.
\end{example}

\subsection{Analysis of Affected Region Reducing}
\label{sec:analysis_topo}
As analyzed, replacing a node $x$ in $G$ with its logically equivalent subgraph $\Delta G$ rooted at $x^{\prime}$ 
triggers level updates in regions $\mathcal{A}_{\mathcal{L}}$ and $\mathcal{A}_{\mathcal{R}}$.
Reducing the size of affected regions is central to efficient logic optimization.
To achieve this, we next exploit the continuous subgraph updates and the dynamic topological order properties to 
reduce the size of $\mathcal{A}_{\mathcal{L}}$  and $\mathcal{A}_{\mathcal{R}}$, and thereby bounding the region by $|\Delta G_i|$.

\minisection{Selective Updates among Continuous Transformation} 
By problem formulation, the circuit $G$ undergoes a sequence of continuous subgraph updates $\Delta G_i$, to the initial network $G_0$.
Each $\Delta G_i$ generates affected regions $\mathcal{A}_{\mathcal{L}}$ and $\mathcal{A}_{\mathcal{R}}$.
Indeed, $\mathcal{A}_{\mathcal{L}}$ and $\mathcal{A}_{\mathcal{R}}$ are already the minimal regions that maintain the level and reverse maps of the entire $V$~\cite{KatrielMH05}.

However, the level constraint specified in~\Cref{eq:level_const} imposes a more focused requirement. 
Specifically, this constraint necessitates the accurate values of only the candidate replacement root's level, $\mathcal{L}(x^{\prime})$, and the original node's reverse level, $\mathcal{R}(x)$, at the precise moment a transformation decision is made. 
This localized requirement motivates a selective update strategy for $\mathcal{L}$ and $\mathcal{R}$, rather than exhaustively recomputing values throughout the entire affected regions $\mathcal{A}_{\mathcal{L}}$ and $\mathcal{A}_{\mathcal{R}}$. 
Thus, by ensuring the correctness of $\mathcal{L}(x^{\prime})$ and $\mathcal{R}(x)$ just prior to applying each update $\Delta G_i$, the global level constraint remains satisfied throughout the sequence of transformations that evolve $G_0$ through $G_1, \ldots, G_{|V|}$ (where $G_i = G_{i-1} \oplus \Delta G_i$).

\begin{property}
\label{prop:selective_update}
When incrementally replacing node $x$ with its logically equivalent node $x^{\prime}$, 
ensuring the correctness of level constraints for $G$ only requires the incremental maintenance of $\mathcal{L}(m)$ for nodes $m$ in the transitive fanin of $x^{\prime}$, and $\mathcal{R}(n)$ for nodes $n$ in the transitive fanout of $x$.  
\end{property}
The level $\mathcal{L}$ and reverse level map $\mathcal{R}$ are computed based on forward and backward transitive path analyses within the graph, respectively. 
\Cref{prop:selective_update} identifies the sufficient scope for incremental updates based on these computational dependencies.
Besides,~\Cref{prop:selective_update} reveals that strict, immediate maintenance of globally correct level map across all potentially affected nodes ($\mathcal{A}_{\mathcal{L}}$ and $\mathcal{A}_{\mathcal{R}}$) is not mandated after each local update. 
This \emph{relaxation of the update} is fundamental, it motivates and enables the design of our efficient incremental algorithm,
which strategically performs selective updates only where necessary, rather than undertaking exhaustive recalculations.
This marks a crucial shift from computationally intensive global updates to an efficient, selective update paradigm. 

An intuitive optimization attempt is to defer updates within the affected region $\mathcal{A}_{\mathcal{L}}$, employing a lazy strategy. 
The actual level recomputation for these deferred regions is triggered only when a subsequent operation requires the level of a node whose calculation (or that of its transitive dependencies) relies on the pending updates within these regions.

However, as shown in~\Cref{fig:lazy_update} by our comparison with \inc~\cite{abc}, this approach yields limited benefits. 
Although the number of traversals slightly decreases (14.5\%), the aggregation of deferred updates in the queue increases the traversal time (45.8\%). 
This marginal reduction in workload offers negligible improvement in total graph update time and minimal impact on overall logic optimization efficiency. 
 
Therefore, the simple deferral mechanism offers insufficient gains. 
It remains to design a more sophisticated approach that explicitly exploits the dynamic AIG's structure changes occurring within the updates.

\minisection{Dynamic Partial Topological Order Maintenance}
To incorporate the structural properties, we analyze the intrinsic relationship between the level and topological order.

\begin{property}
\label{prop:topo_order_level}
A topological order of DAG is a linear extension of the partial order induced by node levels.
\end{property}

In a DAG, if $\mathcal{L}(m) < \mathcal{L}(n)$, then there exist a path from $m$ to $n$, which implies $m$ precedes $n$ in any valid topological order. It naturally extends the level-induced partial order into a total linear order.
Thus, by~\Cref{prop:topo_order_level},  if we can efficiently maintain the order, we gain a structural skeleton that significantly constrains the affected region for level recomputation.  

In a static graph $G$, levels and reverse levels can be computed efficiently using DFS-based or topological order approaches. 
However, in our dynamic setting, where the graph evolves to $G \oplus \Delta G$, these modifications can invalidate the level information for a large portion of the graph. 
As synthesis iteratively applies numerous small updates $\{\Delta G_i\}_{i=1}^{|V|}$, naively recomputing levels by traversing the affected regions after each update (\eg using~\Cref{alg:level_basic}) becomes too costly. 

\begin{figure}[tb!] 
\centering
\begin{subfloat}
    \centering
    \includegraphics[width=.7\linewidth]{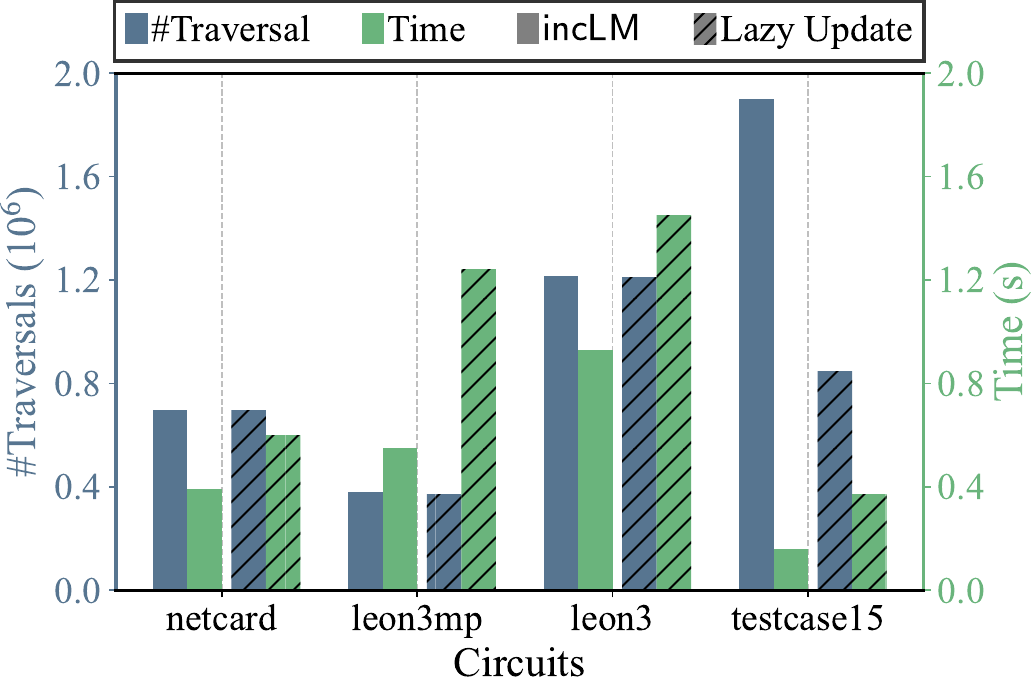}
\end{subfloat} 
\caption{The negligible efficiency gain of lazy level update.} 
\label{fig:lazy_update}   
\end{figure} 

To address this challenge, we directly leverage~\Cref{prop:topo_order_level}.
Instead of discarding all structural information, we formally maintain the \emph{partial topological order} of  $G$.
This maintained order serves as an explicit representation of the structural dependencies between nodes.
When a local modification $\Delta G_i$ occurs, we only need to update this order locally. 
Note that, preserving a complete topological order for each intermediate graph $G_i$ would be computationally excessive. Instead, we focus on the processing order of unhandled nodes by preserving a {partial topological order} on the original node set $V$.

{\emph Remark:} Since we maintain a partial topological order only over unhandled nodes, we prioritize order updates over order tests. 
This enables the use of a simpler data structure while retaining sufficient dynamic capabilities.
Thus, we utilize a \emph{linked list} to balance algorithmic simplicity with dynamic responsiveness, avoiding the implementation complexity of ordered lists~\cite{bender2002two} and the static integer labels~\cite{pk2010batch}.

\section{Bounded Level-Maintained Synthesis Framework}
\label{sec:method_continuous}   
In this section, we present a dynamic bounded level maintenance algorithm specifically designed for local transformation-based synthesis with level constraints.  
Existing dynamic level algorithms~\cite{abc, KatrielMH05, Ramalingam1996dynamic} are unbounded \wrt the local updates $|\Delta G|$, resulting in $O(|V|^2)$ time, causing significant runtime overhead.
In contrast, our algorithm leverages the insights from the analyses to achieve update costs bounded $\|\Delta G\|$ rather than the entire $|V|$ for each subgraph $\Delta G$. 

  
The main result is stated below. 
\begin{theorem}
\label{the:main_complexity}
Given an AIG $G(V, E)$ with maximum allowed level $\mathcal{L}_{\max}$, there exists a bounded dynamic level update algorithm that 
performs a sequence of logic-equivalent subgraph transformations $\{\Delta G_i\}_{i=1}^{|V|}$ to optimize $G$ while maintaining the level constraint 
in $O(|V| \Delta\log\Delta)$ total time, where $\Delta = \max_i \|\Delta G_i\|$ is the maximum extended size of $\Delta G_i$.
\end{theorem}

\subsection{Overview of Bounded Algorithm} 
Building on these foundations, we now provide an overview of our level-constrained synthesis framework~\incour enabled by bounded dynamic level computation, as outlined in Algorithm~\ref{alg:overall}.
The general idea of our~\incour is to maintain the partial topological order to reduce the affected regions $\mathcal{A}_\mathcal{L}$ and $\mathcal{A}_\mathcal{R}$,  thereby bounded by $ \Delta G_i $. 

 

\begin{algorithm}[tb!]
    \small 
    \caption{\small Bounded dynamic level constrained synthesis~\incour} \label{alg:overall} 
    \KwIn{AIG $G(V,E)$, maximum allowed level $\mathcal{L}_{\max}$ } 
    \KwOut{Optimized AIG $G$, final max level $\mathcal{L}^{\prime}_{\max}$}  
    Compute topological order $\mathcal{T}$ of $V$\;
    Compute level $\mathcal{L}$ and reverse level map $\mathcal{R}$ of $V$\; 
    Initialize  handle status $handle$ for nodes in $G$\;                    
    \ForEach{ \textnormal{node} $x$ \textnormal{in partial topological order} $\mathcal{T}$}{
        Set $handle[x]$ as true\;
        $\mathcal{T} \gets$ $\mathcal{T}  \setminus x$\;
        Compute $\mathcal{L}$ of $x$ by Alg.~\ref{alg:dynamic_level}~\dynlev\;
        Attempt local synthesis transformation at node $x$\;
        $\Delta G \gets$ replacement subgraph rooted at $x^{\prime}$\; 
        \lIf{$\mathcal{L}(x^{\prime}) > \mathcal{L}_{\max} - \mathcal{R}(x)$}{\textbf{continue}}
        Apply $\Delta G_i$ to $G_{i-1}$, \ie $G_{i} = G_{i-1} \oplus \Delta G_i$, yielding affected node sets $\mathcal{I}_{x,x^{\prime}}$, $\mathcal{P}_{x,x^{\prime}}$, $\mathcal{D}_{x,x^{\prime}}$\;
        Compute  $\mathcal{L}$ of $\mathcal{I}_{x,x^{\prime}}$~by Alg.~\ref{alg:dynamic_level}~\dynlev\; 
        Maintain topological order $\mathcal{T}$ by Alg.~\ref{alg:dynamic_order}~\dynto\; 
        Compute  $\mathcal{R}$ of $ \{x^{\prime}\} \cup \mathcal{D}_{x,x^{\prime}}$ by Alg.~\ref{alg:dynamic_reverse}~\dynrl\;
    }
    Compute the final maximum level  $\mathcal{L}^{\prime}_{\max}$\;
    \Return{\textnormal{Optimized} $G_{|V|}$, $\mathcal{L}'_{\max}$\;}
\end{algorithm}

Specifically, Algorithm~\incour takes as input an AIG $G(V,E)$ and a maximum allowed level $\mathcal{L}_{\max}$, and return the optimized $G$ and the new maximum level $\mathcal{L}^{\prime}_{\max}$ with $\mathcal{L}^{\prime}_{\max} \leq \mathcal{L}_{\max}$. 
It first computes the topological order $\mathcal{T}$ for nodes $V$, along with their corresponding level $\mathcal{L}$ and reverse level map $\mathcal{R}$.  
A handle status, $handle$, is initialized for $V$, to prune affected regions and maintain the partial topological order (lines 1--3).
The core of framework~\incour iterates each node $x$ of original $V$ according to the order $\mathcal{T}$ by~\Cref{prop:topo_order_level} (lines 4--14). 
For each node $x$:
\begin{enumerate}
    \item The handle status $handle[x]$ is set to true, and partial topological order $\mathcal{T}$ is updated by removing $x$ to form a new order for the remaining nodes (lines~5,~6).
    \item The level map $\mathcal{L}$ of $x$ is updated by Algorithm~\ref{alg:dynamic_level}~\dynlev, 
    thereby preserving the correctness of level computations for node in the transitive fanout of $x$ during subsequent processing phases, as established in~\Cref{prop:selective_update}. The details of \dynlev are found in~\Cref{sec:method_level} (line~7).
    \item  A local synthesis transformation is attempted at node $x$, \eg \rw and \rf, optimizing the local logic, often aiming to reduce node count or area.  This attempt may identify a candidate replacement subgraph $\Delta G_i$ (where $i$ can be seen as an iteration index \wrt node $x$) that is locally equivalent to the logic driven by $x$. This $\Delta G_i$ is rooted at a new or existing node $x^{\prime}$ (lines~8--10). 
    \item The transformation is only accepted when $\mathcal{L}(x^{\prime})$ do not exceed the available level budget by~\Cref{eq:level_const}. 
    The graph $G_{i-1}$ is updated to $G_i$ by incorporating the operation \kw{insert}, \kw{delete}, and \kw{replace} from $\Delta G_i$, \ie $G_{i} = G_{i-1} \oplus \Delta G_i$.   It also identifies the directly affected node sets resulting from the replacement $x$ by $x^{\prime}$, \ie $\mathcal{I}_{x,x^{\prime}}$, $\mathcal{P}_{x,x^{\prime}}$, and $\mathcal{D}_{x,x^{\prime}}$ (line~11).
    \item Following the valid order, \Cref{alg:dynamic_level} \dynlev updates the levels $\mathcal{L}$ for newly inserted nodes  $\mathcal{I}_{x,x^{\prime}} $ (line~12). 
    \item Since \kw{replace} may invalidate the partial topological order $\mathcal{T}$, \Cref{alg:dynamic_order}~\dynto maintains $\mathcal{T}$ from $x^{\prime}$, pruning the traversal by $handle$, as detailed in~\Cref{sec:method_topo} (line~13). 
    \item Finally, the reverse levels for nodes affected by the transformation, starting from $ \{x^{\prime}\} \cup \mathcal{D}_{x,x^{\prime}}$, are updated and pruned by $handle$ using algorithm~\ref{alg:dynamic_reverse}~\dynrl, described in~\Cref{sec:method_reverse} (line~14). 
 \end{enumerate}   
After iterating through all nodes in the dynamic-maintained order, the final maximum level $\mathcal{L}^{\prime}_{\max}$ of the resultant AIG $G_{|V|}$ is computed (lines~15, 16).

\section{Dynamic Maintenance for Single Subgraph Update}
\label{sec:method_single}  
Building upon~\incour, designed to manage continuous subgraph updates, iteratively handles single subgraph updates while adhering to level constraints.
We next present three key dynamic components for single updates \ie dynamic partial topological order maintenance to efficiently update structure evolution in~\Cref{sec:method_topo}, dynamic level computation to selectively update node levels only within the necessary affected regions in~\Cref{sec:method_level}, and dynamic reverse level computation to prune recomputation traversals in~\Cref{sec:method_reverse}.

\subsection{Dynamic Partial Topological Order Maintenance}
\label{sec:method_topo}  
We first present our dynamic algorithm~\dynto, which maintains partial topological order throughout graph updates.
This order directly benefits both level and reverse level computations,  based on~\Cref{prop:topo_order_level} in~\Cref{sec:analysis}.

A partial topological order $\mathcal{T}$ over unhandled nodes becomes locally invalid if a graph transformation $\Delta G_i$ introduces new fanin-fanout dependencies among these unhandled nodes that violate their current sequence in $\mathcal{T}$. 
For instance, if an unhandled node $m$ acquires a new fanin $f$ (also unhandled) due to $\Delta G_i$, but $f$ currently appears after $m$ in $\mathcal{T}$ (\ie $m \preccurlyeq f$), then $\mathcal{T}$ is no longer a valid order. Algorithm~\dynto addresses these invalidations caused by \kw{replace} edges.

\begin{algorithm}[tb!]
\small
    \caption{\small Dynamic order maintenance~\dynto} \label{alg:dynamic_order} 
    \KwIn{AIG $G_i$, partial topological order $\mathcal{T}$, handle status $handle$, resynthesis nodes $x$, $x^{\prime}$ } 
    \KwOut{Update partial topological order $\mathcal{T}$} 
    \SetKwComment{tcp}{\footnotesize // }{}
     
    Initialize $visit[n] \gets \textnormal{false}$ for all nodes $n \in V$\;
    $inv \gets \emptyset$ \tcp*{Nodes with invalid orders}
    \find{($x^{\prime}, \ visit, \ handle, \ inv, \ G_i$)}\;
    \If{$handle[x^{\prime}] = \textnormal{false}$} {
        Append $x^{\prime}$ to $inv$\;
    }
    \Return{$\reorder(x, \ inv, \ \mathcal{T})$\;} 
    \BlankLine
    \SetKwProg{Fn}{Procedure}{:}{}
    \Fn{\find{($n, \ visit, \ handle, \ inv$)}}{
        \ForEach{ \textnormal{fanin node} $f$ \textnormal{of} $n$ }{
            \lIf{$handle[f]$}{\KwSty{continue}} 
            \lIf{$visit[f]$}{\KwSty{continue}} 
             $visit[f] \gets \textnormal{true}$\;
             $\find(f, vis, han, inv)$\;
             Append $f$ to $inv$\;
        }  
    }  
    \BlankLine
    \SetKwProg{Fn}{Procedure}{:}{}
    \Fn{\reorder{($x, inv, \mathcal{T}$)}}{
        $curOrd \gets \textnormal{the order element of } x \textnormal{ from } \mathcal{T}$\;
        \ForEach{ \textnormal{node} $f$ in $inv$ }{
            $newOrd \gets \mathcal{T}. \insertafter(curOrd, f)$\;
            Set $newOrd$ to the order element of  $f$ of $\mathcal{T}$\;
            $curOrd \gets newOrd$\;
        }
       \Return{\textnormal{Updated order} $\mathcal{T}$\;}
    } 
\end{algorithm}

As shown in~\Cref{alg:dynamic_order}, \dynto finds the nodes with invalid order (procedure~\find) and reorder these nodes (procedure~\reorder) to maintain the partial topological order $\mathcal{T}$.  
It takes as input the current AIG $G_i$, the partial topological order $\mathcal{T}$ (represented as a linked list of unhandled nodes), the $handle$ status array,  the original resynthesis nodes $x$ and $x^{\prime}$, and returns the updated order $\mathcal{T}$.
Note that, if $x'$ has been handled,  \dynto does nothing since the order $\mathcal{T}$ remains valid. 

Algorithm \dynto first initialize $visit$ to false for nodes in $V$, the global array $inv$ to $\emptyset$ to store nodes with invalid orders (lines 1, 2).
It then invokes procedure \find on the resynthesis node $x^{\prime}$ to discover nodes with invalid partial topological orders, and $x^{\prime}$ is further appended to $inv$ when $x^{\prime}$ is unhandled (lines 3--5).
Finally, procedure \reorder is called to restore the valid partial topological order using the collected invalid nodes in $inv$ (line 6).


Procedure \find discovers the nodes with invalid partial topological orders (stored in $inv$) by backward DFS from $x^{\prime}$ on $\Delta G_i$. 
It accumulates all encountered unhandled and unvisited fanins into $inv$ in a post-order. This ensures $inv$ contains a topologically sorted sequence of all unhandled nodes within the fanin cone of $x^{\prime}$ (lines 8--13).
Unhandled node $x^{\prime}$ itself is appended to $x$ after its fanin cone is explored (line 4, 5).
 
Procedure \reorder is then called to restore a valid partial topological order.  It uses the original resynthesis node $x$ as an anchor point ($curOrd$) in the linked list $\mathcal{T}$. The new unhandled node's order must immediately follow $curOrd$ in the updated partial topological order of $\mathcal{T} \gets \mathcal{T} \setminus x$ (line 17).
It iterates through the nodes $f$ in $inv$, which are already topologically sorted relative to each other (lines 16--19). 
Each node $f$ is removed from its current position in $\mathcal{T}$, and re-inserted into $\mathcal{T}$ immediately after $curOrd$. 
Then,  $curOrd$ is updated to the newly positioned order of $f$. 
It effectively relocates all nodes from $inv$ into $\mathcal{T}$ to dynamically maintain the partial topological order $\mathcal{T}$.


We illustrate~\Cref{alg:dynamic_order} with an example below.
\begin{figure}[tb!]
    \centering
    \quad 
    \includegraphics[width=.80\columnwidth]{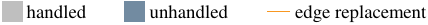}
    \newline 
    \subfloat[Before $\Delta G_i$ update. \label{subfig:rw_before}]{
        \begin{minipage}{0.3 \columnwidth}   
            \centering                           
            \includegraphics[width=0.9\columnwidth]{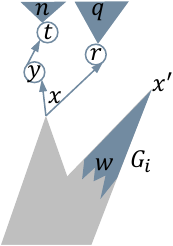}
        \end{minipage} 
    } 
    \quad\quad\quad
    \subfloat[After $\Delta G_i$ update. \label{subfig:rw_after}]{
        \begin{minipage}{0.3 \columnwidth}   
            \centering                           
            \includegraphics[width=0.9\columnwidth]{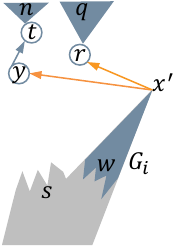}
        \end{minipage} 
    }  
    
    \caption{Example of partial topological order maintenance for replacing $x$ by $x^{\prime}$.}
    \label{fig:example_rw}
\end{figure}

\begin{example}
The example in~\Cref{fig:example_rw} depicts partial topological order maintenance following the replacement of $x$ by $x'$, \ie $x$'s fanouts are redirected to that of $x^{\prime}$.
Note that, we use ``$\cdots$'' to represent nodes within the region.

Initially, assume that the order $\mathcal{T}$ in~\Cref{subfig:rw_before} is  $\mathcal{T} = \{x \preccurlyeq y \preccurlyeq t \preccurlyeq \ldots \preccurlyeq n \preccurlyeq r \preccurlyeq \ldots \preccurlyeq  q \preccurlyeq   w \preccurlyeq  \ldots \preccurlyeq x^{\prime} \}$ in~\Cref{subfig:rw_before}.
However, after $G_{i-1} \oplus \Delta G_i$, $\mathcal{T}$ becomes invalid because the replacement creates a path from $x^{\prime}$ to $y$, requiring $x^{\prime} \preccurlyeq y$ in any valid order, as shown in~\Cref{subfig:rw_after}.
To restore validity,  procedure \find is called on unhandled node $x^{\prime}$, and identifies $inv = \{w, \ldots, x^{\prime}\}$. 
Procedure \reorder places nodes of $inv$ topologically after $x$ in $\mathcal{T}$, creating a new valid order $\mathcal{T} = \{w \preccurlyeq  \ldots \preccurlyeq x^{\prime}  \preccurlyeq y \preccurlyeq t \preccurlyeq \ldots \preccurlyeq n \preccurlyeq r \preccurlyeq \ldots \preccurlyeq  q \}$.
Note that, $x$ is removed from the original $\mathcal{T}$ as it has been handled.
\end{example}

The correctness of \dynto is assured as follows.


\begin{theorem}
\label{the:topo_order_correct}
After each transformation $\Delta G_i$ applied by Algorithm~\incour, Algorithm~\dynto 
correctly updates the sequence of unhandled nodes $\mathcal{T}$ such that it remains a valid partial topological order.
\end{theorem}


\begin{proof}
We show this by a loop invariant.

\emph{Loop invariant:} Before the start of each update $G_{i-1} \oplus \Delta G_i$, Algorithm~\dynto maintains the partial topological order for each $i$ from 1 to $|V|$.  

For iteration $i = 1$, it is easy to verify that the loop invariant holds from line 1 in Algorithm~\incour. 
Assume that the loop invariant holds for $i > 1$, and we show the loop invariant holds for $i + 1$, \ie for any two unhandled nodes $m, n$ in $\mathcal{T}_{i+1}$, if there is a path $m \leadsto n$ in $G_{i+1}$, then $m \preccurlyeq n$ in $\mathcal{T}_{i+1}$.

1) Neither $m$ nor $n$ is in $inv$. 
    Their relative order in $\mathcal{T}_{i+1}$ is the same as in $\mathcal{T}_{i}$, and remains correct as the transformation $G_{i} \oplus \Delta G_{i+1}$ must not have created a new path $m \leadsto n$ that inverts their old order.
2) Both $m$ and $n$ are in  $inv$.
    Since $inv$ is created by \find in a post-order traversal of fanins, $inv$ itself is topologically sorted~\cite{cormen2022introduction}. 
    Procedure \reorder preserves this relative order when inserting into $\mathcal{T}_{i+1}$. 
    Thus, $m \preccurlyeq n$ holds in $\mathcal{T}_{i+1}$.
3) One of $m, n$ is in $inv$, the other is not. 
    We assume $m \in inv$, $n \notin inv$. 
    If $n$ was originally after $curOrd$ and not in $inv$, it will remain after $m$ in $\mathcal{T}_{i+1}$, \ie $m \preccurlyeq n$. 
    This is because all nodes in $inv$ are reinserted immediately after the order of $x$ (\ie $curOrd$).
    If $n$ was originally before $curOrd$, and not in $inv$, then a path $m \leadsto n$ in $G_{i+1}$ introduces a contradiction.
    As $m$ must be handled based on loop invariant $i$, \ie $m$ is handled and not in $inv$ as \find prunes all handles node, contradicting with $m \in inv$ (lines 9, 10 in \dynto).
    We can similarly prove for the case $m \notin inv$, $n \in inv$.   

This shows that the loop invariant holds for $G_{i+1}$.

Putting these together, and we have the conclusion.
\end{proof}

The time complexity of \dynto is analyzed as follows.
\begin{theorem}
\label{the:topo_order_complexity}
Algorithm~\dynto maintains the partial topological order $\mathcal{T}$ in $O(|V| \Delta)$ time for continuous updates  $\{\Delta G_i\}_{i=1}^{|V|}$, where $\Delta = \max_i \|\Delta G_i\|$ is the maximum extended size of $\Delta G_i$. 
\end{theorem} 
\begin{proof}

The time complexity of \dynto for single $\Delta G_i$   is determined by the size of nodes with invalid orders, \ie $inv$. This is because both \find and \reorder procedures perform a one-pass traversal over the nodes in $inv$. 
Thus, the complexity for maintaining the order after a single $\Delta G_i$ is $O(|inv|)$. 
To determine the overall complexity, we establish an upper bound on $|inv|$ from the structural properties of the resynthesis process and partial topological order maintenance.

The resynthesis at node $x$ is based on a $k$-feasible cut that represents the logic function of $x$.
When replacing $x$ by $x'$, the backward traversal in procedure \find from $x'$ is confined within the fanin cone defined by this $k$-feasible cut. 
This is because $x'$ implements the same logic as the original cut of $x$.
Besides, by the definition of $k$-feasible cuts, every cut node $n$ has a path to $x$ (\ie $n \preccurlyeq x$ in $\mathcal{T}$). 
Since Algorithm~\incour processes nodes following the order by~\Cref{the:topo_order_correct}, and $x$ is the current node being handled, all cut nodes $n$ must have been handled already due to $n \preccurlyeq x$.

Therefore, when \find performs backward traversal from $x'$, it encounters these already-handled cut nodes and terminates the traversal (line 9 in Algorithm~\dynto). 
This bounds the size of invalid nodes: $|inv| \leq |\Delta G_i| - k$.

Algorithm~\incour performs $|V|$ resynthesis steps, invoking \dynto for each.
Let $\Delta = \max_i \|\Delta G_i\|$ be the maximum extended size of $\Delta G_i$. 
The total complexity for $\{\Delta G_i\}_{i=1}^{|V|}$ is:
$$\sum_{i=1}^{|V|} O(|inv_i|) \leq \sum_{i=1}^{|V|} O(|\Delta G_i| - k) \leq O(|V| \Delta)$$ 

Putting these together, and we have the conclusion.
\end{proof}


\begin{algorithm}[tb!]
\small
    \caption{\small Dynamic level computation~\dynlev} \label{alg:dynamic_level} 
    \KwIn{AIG $G_i$, level map $\mathcal{L}$, nodes to be updated $U_n$ } 
    \KwOut{Updated level map  $\mathcal{L}$ } 
    \ForEach{\textnormal{node} $n \in U_n$}{
        $\mathcal{L}(n) \gets 1 + \max\{\mathcal{L}(f) \mid f \in \textnormal{fanin}(n)\}$\;
    }
    \Return{\textnormal{Updated level} $\mathcal{L}$\;}
\end{algorithm}

\subsection{Dynamic Level Computation}
\label{sec:method_level}  
We next introduce our dynamic level computation method~\dynlev, which efficiently handles selective updates during continuous subgraph transformations, leveraging the maintained partial topological order.

Based on~\Cref{prop:selective_update}~\Cref{prop:topo_order_level} in~\Cref{sec:analysis}, the order serves as a structural dependency that captures the dynamic graph evolution, constraining the affected regions.
Moreover, by~\Cref{eq:level_const}, only selective level updates are needed to maintain the correct level when replacement node $x^{\prime}$.

As shown in~\Cref{alg:dynamic_level}, \dynlev updates level values for a specified set of nodes. 
It takes as input the current AIG $G_i$, the level map $\mathcal{L}$, and the node set $U_n$ requiring level updates, and returns the updated level map $\mathcal{L}$. 
The computation follows the standard level definition, where each node's level as one plus the maximum level among its fanin nodes. 
In the context of Algorithm~\incour, $U_n$ represents different node sets depending on the invocation: the handling node $x$ (line 7), or the set $\mathcal{I}_{x,x'}$ of newly inserted nodes (line 12). 
Note that, the nodes in $\mathcal{I}_{x,x'}$ naturally maintain proper dependency since they are constructed through forward fanin exploration.

The correctness of \dynlev is assured as follows. 
\begin{theorem}
\label{the:level_correct}
For each candidate replacing $x$ by $x^{\prime}$ in current AIG $G_i$, Algorithm~\dynlev correctly computes $\mathcal{L}(x^{\prime})$.  
\end{theorem}

\begin{proof}
We show this by~\Cref{the:topo_order_correct} and loop invariant.

By~\Cref{the:topo_order_correct}, when processing node $x$, all nodes in its transitive fanin cone have been processed and assigned correct level values. Since the replacement node $x^{\prime}$ is derived from a $k$-feasible cut rooted at $x$, all nodes in the cut have already been processed.
The level $\mathcal{L}(x^{\prime})$ is computed topologically based on the newly created nodes $\mathcal{I}_{x,x^{\prime}}$, which grow from the handled cut nodes, using the standard level definition.
That is, as long as the handled nodes have correct levels, then $\mathcal{L}(x^{\prime})$ is computed correctly.

Thus, to establish that~\Cref{the:level_correct} holds, we prove the loop invariant: for any node $x$ with ${handle}[x] = \text{true}$, its level $\mathcal{L}(x)$ is correct \wrt the current graph $G_i$.

Since nodes are processed in topological order, the fanins of  $x^{\prime}$ have been processed with correct levels, as guaranteed by Theorem~\ref{the:topo_order_correct}. Thus, the loop invariant holds, ensuring the correctness of the level computations throughout the algorithm.
 
Putting these together, and we have the conclusion. 
\end{proof}

\begin{algorithm}[tb!]
\small
    \caption{\small Dynamic reverse level computation~\dynrl } \label{alg:dynamic_reverse} 
    \KwIn{AIG $G_i$, level map $\mathcal{R} $, handle status $handle$   } 
    \KwOut{Updated reverse level map  $\mathcal{R}$ } 
    Initialize $visit[n] \gets \textnormal{false}$ for all nodes $n \in V$\;
    Initialize $queue \gets \emptyset$\;
    Push the starting node of $\{x'\} \cup \mathcal{D}_{x,x'}$ into $queue$\; 
    \While{$queue \neq \emptyset$}{
        $n \gets \textnormal{pop from } queue$\; 
        \lIf{$visit[n]$}{\KwSty{continue}}
        \ForEach{\textnormal{node} $m \in \textnormal{fanin}(n)$}{ 
            \lIf{$handle[m]$} {\KwSty{continue}}
            $r_{\textnormal{new}} \gets 1 + \max\{\mathcal{R}(f) \mid f \in \textnormal{fanout}(m)\}$\;
            \lIf{$\mathcal{R}(m) = r_{\textnormal{new}}$} {\KwSty{continue}}
            $\mathcal{R}(m) \gets r_{\textnormal{new}}$\;
            \If{ $! visit[m]$ }{
                $visit[m] \gets$ \textnormal{true}\;
                Push $m$ into $queue$\;
            }
        } 
    }  
    \Return{\textnormal{Updated reverse level} $\mathcal{R}$\;}
\end{algorithm}

The time complexity of \dynlev is analyzed as follows.

\begin{theorem}
\label{the:level_complexity}
Algorithm~\dynlev computes the level map $\mathcal{L}$ in $O(|V|)$ total time for continuous update  $\{\Delta G_i\}_{i=1}^{|V|}$.
\end{theorem}  

\begin{proof}
The time complexity of \dynlev arises from dealing with individual nodes and the directly affected nodes by \kw{insert}.

1) For each original node $x \in V$, Algorithm~\dynlev is invoked once to compute $\mathcal{L}(x)$  (line 7 in~\Cref{alg:overall}) . 
Since each node has two fanins in AIG, the total cost for processing all individual nodes is $O(2\cdot|V|)$ time.

2) For newly inserted nodes $\mathcal{I}_{x,x'}$, local transformation-based synthesis flow is heuristic applied only when the subgraph size $|\Delta G_i|$ is reduced (line 8 in~\Cref{alg:overall}). Therefore, $\sum_{i=1}^{|V|} |\mathcal{I}_{x,x'}|_i < |V|$,  yielding total cost $O(|V|)$ time. 

Putting these together, and we have the conclusion.
\end{proof}

Benefiting from the maintained partial topological order,  \Cref{the:level_correct} and \Cref{the:level_complexity} tell us that level-constrained synthesis under continuous subgraph updates can avoid costly level propagation across the affected region $\mathcal{A}_{\mathcal{L}}$. 
Instead, level computation is only required for the newly created nodes and for the nodes followed by partial topological order.

\subsection{Dynamic Reverse Level Computation}
\label{sec:method_reverse}  
We present our dynamic reverse level computation method \dynrl, which efficiently computes the level budget constraints while leveraging the maintained partial topological order.

The reverse level of a candidate node $x$ constrains the maximum budget of its replacement $x^{\prime}$, ensuring that the optimized level does not exceed $\mathcal{L}_{\max}$. 
Based on~\Cref{prop:selective_update}~\Cref{prop:topo_order_level} in~\Cref{sec:analysis}, following the valid order, we can further prune the handled nodes during the exhaustive backpropagation to fanins when updating reverse levels. 
That is, \dynrl mainly computes $\mathcal{R}$ for these newly inserted nodes in $\Delta G_i$.

\begin{figure}[tb!]
    \centering
    \quad 
    \includegraphics[width=.80\columnwidth]{figs/example_topo_legend.pdf}
    \newline 
    \subfloat[Before $\Delta G_i$ update. \label{subfig:reverse_before}]{
		\includegraphics[width=0.32\columnwidth]{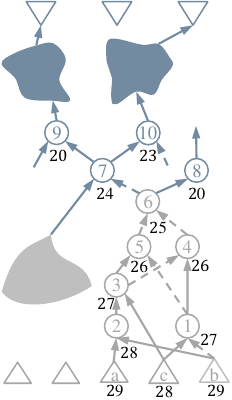}}
    \quad\quad\quad
    \subfloat[After $\Delta G_i$ update.  \label{subfig:reverse_after}]{
		\includegraphics[width=0.32\columnwidth]{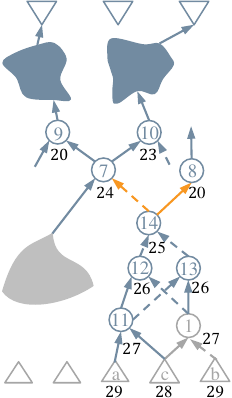}}  
    \caption{Example of reverse level $\mathcal{R}$ computation for replacing $6$ by $14$. The $\mathcal{R}$ of each node is annotated below its node.}
    \label{fig:example_reverse}
\end{figure}

As shown in~\Cref{alg:dynamic_reverse}, \dynrl updates reverse levels after a graph transformation $\Delta G_i$ has been applied. 
It takes as input the current AIG $G_i$, the reverse level map $\mathcal{R}$ and the $handle$ status, and returns the updated $\mathcal{R}$. 
Algorithm~\dynrl first initializes an array $visit$ and a minimum priority queue $queue$ for a backward traversal. The starting nodes in $\{x'\} \cup \mathcal{D}_{x,x'}$ are pushed into the $queue$ (lines 1--3).  
Algorithm~\dynrl iteratively processes to update their reverse levels or trigger updates for their unhandled fanins (lines 4--14).
Specifically, for each popped and currently handling node $n$, it examines each fanin node $m$ of $n$. 
The reverse level only persists updates in the unhandled transitive fanin cone, \ie it is skipped when $m$ has already been handled (lines 5--8). 
The new potential reverse level $r_{\textnormal{new}}$ for $m$ is calculated using the standard definition: $1 + \max\{\mathcal{R}(f) \mid f \in \textnormal{fanout}(m)\}$ (line 9). 
If $r_{\textnormal{new}}$ is different from the old $\mathcal{R}(m)$, then $\mathcal{R}(m)$ is updated to $r_{\textnormal{new}}$ (lines 10, 11). 
If $m$'s reverse level changed and $m$ has not been visited, $m$ and its reverse level $\mathcal{R}(m)$ are pushed into $queue$ for backward propagation (lines 12--14). 
Note that, different from the affected region $\mathcal{A}_\mathcal{R}$, the starting nodes for \dynrl are those with inconsistent reverse level values in $\{x'\}$, since node in $\mathcal{D}_{x,x'}$ has been handled by~\Cref{the:topo_order_correct}.
 
 
By handling nodes in $V$ by the partial topological order, \dynrl effectively constrains recomputation to nodes within $\Delta G_i$ and their immediate neighbors. 
This localized approach provides substantial benefits when transformations occur near POs, where \inc requires traversing almost the entire $V$.

We illustrate~\Cref{alg:dynamic_reverse} with an example below.

\begin{example}
\Cref{fig:example_reverse} illustrates reverse level computation~\dynrl for replacing node 6 by 14.
The update process is initiated from $\{14\}$,  as  $\mathcal{D}_{6,14} = \{1, c, a, b\}$ has been handled.

Algorithm~\dynlev first updates $\mathcal{R}(14) = 25$, then triggers backward propagation to affected $\{12, 13, 11\}$ based on their reverse levels as priority.
Thus,  $\{12, 13, 11\}$ are sequentially dequeued from $queue$ and assigned updated reverse levels: $\mathcal{R}(12) = 26$, $\mathcal{R}(13) = 26$, and $\mathcal{R}(11) = 27$.
\end{example}

The correctness of \dynrl is assured as follows. 
\begin{theorem}
\label{the:reverse_correct}
For each candidate replacing $x$ by $x^{\prime}$ in current AIG $G_i$, Algorithm~\dynrl correctly computes $\mathcal{R}(x)$.  
\end{theorem}

\begin{proof}
We present a proof sketch using loop invariants.

By~\Cref{prop:selective_update} and~\Cref{the:topo_order_correct}, it is sufficient to maintain the correct reverse levels only for unhandled nodes. 
Algorithm \dynrl ensures this by pruning backward propagation at any handled node $m$ (line 8). 
Correctness follows similarly with~\Cref{the:topo_order_correct,the:level_correct}.
\end{proof}

The time complexity of \dynrl is analyzed as follows.

\begin{theorem}
\label{the:reverse_complexity}
Algorithm~\dynrl computes the level map $\mathcal{L}$ in $O(|V| \Delta \log \Delta )$ time for continuous update  $\{\Delta G_i\}_{i=1}^{|V|}$, where $\Delta = \max_i \|\Delta G_i\|$ is the maximum extended size of $\Delta G_i$. 
\end{theorem}

\begin{proof}
For each $\Delta G_i$, Algorithm~\dynrl updates the reverse levels of unhandled nodes within the logical cone corresponding to the new node $x'$. 
As established in the proof of~\Cref{the:topo_order_complexity}, this affected region is bounded by the extended size $\|\Delta G_i\|$, which includes nodes derived from the $k$-feasible cut along with their immediate neighbors. The number of nodes requiring reverse level updates is $|\Delta G_i| - k \leq \|\Delta G_i\|$. 

During the updates, \dynrl traverses the fanouts of affected nodes to derive updated reverse level values (line 9). 
\dynrl maintains a priority queue containing at most $|\Delta G_i| \leq \|\Delta G_i\| $ nodes, where each update requires $O(\|\Delta G_i\| \log \|\Delta G_i\|)$.
 
Since $\|\Delta G_i\| \leq \Delta$, each update takes at most $O(\Delta \log \Delta)$ time, thereby running $O(|V| \Delta \log \Delta)$ for $|V|$ updates.

Putting these together, and we have the conclusion.
\end{proof}

By Algorithm~\incour, we can now establish~\Cref{the:main_complexity}.

\begin{proof}[Proof of Theorem~\ref{the:main_complexity}]
The proof proceeds in two parts: correctness and complexity.

First, the level map $\mathcal{L}$ and reverse level map $\mathcal{R}$ required in~\Cref{eq:level_const} are correctly computed by Algorithms~\dynto,~\dynlev, and~\dynrl,  from~\Cref{the:topo_order_correct,the:level_correct,the:reverse_correct}. 
Therefore, the level constraints of AIG $G$ are satisfied.

Second, Algorithm~\incour maintains level constraints in $O(|V| \Delta \log \Delta)$ time, since \dynto,~\dynlev, and~\dynrl run in $O(|V| \Delta)$, $O(|V| )$, and $O(|V| \Delta \log \Delta)$, respectively, as proven in~\Cref{the:topo_order_complexity,the:level_complexity,the:reverse_complexity}.

Putting these together, and we have the conclusion.
\end{proof}

\emph{Remarks:}
The boundedness of Algorithm~\incour tells us we guarantee that \incour is always more efficient than others, \eg \inc and \incpq, especially when $\Delta G$ is small and $G$ is big.
We maintain levels only through localized updates rather than exhaustive traversal of the entire $G$~\cite{KatrielMH05, abc, Ramalingam1996dynamic}.
In practice,  $\Delta G$ in local transformation-based synthesis is very small, typically with fewer than 10 nodes, leading to $\Delta$ on the order of tens of nodes, \eg 4-feasible cut \rw. 
Under such conditions, $\Delta$ can be treated as a practical constant, making \incour effectively linear in $|V|$.
  
\section{Experimental Result}
\label{sec:exp}
\subsection{Experimental Setting} 
We conduct all experiments on an Intel(R) Xeon(R) Gold 6252 CPU 2.10GHz, 128GB RAM, and an Ubuntu 18.04 system.
All tests are repeated over 3 times and the average is reported. 
The proposed algorithms is integrated into the widely used logic synthesis tool ABC~\cite{abc} to evaluate its effectiveness in logic optimization. 

\begin{table}[tb!]
\centering
\caption{Dataset summarization}
\label{tab:datasets}
\setlength{\tabcolsep}{8pt} 
\resizebox{.9\columnwidth}{!} {
\begin{tabular}{l|rrrr} \toprule
Circuit       & \#Input & \#Output & \#AND      & Level \\ \midrule
hyp~\cite{epfl}         & 256     & 128      & 214,335    & 24,801  \\
mult256~\cite{abc}     & 512     & 512      & 724,133    & 2,377   \\
netcard~\cite{albrecht2005iwls}     & 195,730 & 97,805   & 802,919    & 39      \\
sort1024~\cite{abc}    & 1,024   & 1,024    & 2,773,507  & 2,661   \\
mult512~\cite{abc}     & 1,024   & 1,024    & 2,905,935  & 4,961   \\
sort2048~\cite{abc}    & 2,048   & 2,048    & 10,769,494 & 5,528   \\
mult1024~\cite{abc}    & 2,048   & 2,048    & 11,695,025 & 9,917   \\
sixteen~\cite{epfl}     & 117     & 50       & 16,216,836 & 140     \\
twenty~\cite{epfl}      & 137     & 60       & 20,732,893 & 162     \\
twentythree~\cite{epfl} & 153     & 68       & 23,339,737 & 176     \\
sort4096~\cite{abc}    & 4,096   & 4,096    & 41,920,515 & 13,924  \\ \bottomrule
\end{tabular}
}
\end{table}

\minisection{Dataset}
\Cref{tab:datasets} provides an overview of the circuits used in this study, sourced from the EPFL combinational benchmark suite~\cite{epfl}, IWLS 2005~\cite{albrecht2005iwls}, and an AIG generator from ABC~\cite{abc}. To evaluate the scalability of various level-constrained logic optimization algorithms, the AIGs vary significantly in size, ranging from 214 hundred to 42 million nodes. Note that, AIGs produced by the AIG generator are typically compact and resistant to further optimization. 
To facilitate optimization in our experiments, we utilize a modified \rw operator designed with negative node gain. 
This operator is applied to 10\% of randomly selected nodes, thereby enabling optimization of these initially compact AIGs.

\minisection{Baseline}
We compare our bound dynamic level computation~\incour with~\inc~\cite{abc} and~\incpq~\cite{KatrielMH05}. 
Algorithm~\inc~\cite{abc} employs a predefined level-indexed vector to store candidate nodes at each level. It propagates updates throughout the entire affected regions of $\mathcal{A}_\mathcal{L}$ and $\mathcal{A}_\mathcal{R}$. Each individual update $\Delta G_i$ requires $O(\max(\mathcal{L}_{\max}, \|\aff\|))$ time, with the overall complexity in  $O(|V|^2)$ for continuous updates $\{\Delta G_i\}_{i=1}^{|V|}$. It serves as the default implementation strategy in ABC.
Algorithm~\incpq~\cite{KatrielMH05} utilizes a priority queue to maintain nodes requiring level updates, enabling efficient identification of $\mathcal{A}_\mathcal{L}$ and $\mathcal{A}_\mathcal{R}$. Each individual update $\Delta G_i$ requires $O(\|\aff\| + |\aff\log{\aff})$ time, with the overall complexity in $O(|V|^2\log{|V|})$ for the sequence of updates $\{\Delta G_i\}_{i=1}^{|V|}$.  
Note that, the correctness of optimized circuits is verified through combinational equivalence checking~\cite{abc}.

To evaluate the practical impact of our approach, we extend all three algorithms (\inc, \incpq, \incour) to cut-based local optimization operators: \rw~\cite{alan2006dag} and \rf~\cite{alan2006dag,amaru2018improve}. 
Moreover, to assess the scalability of our algorithm, we further apply it to \rs(in short \rss), which performs window-based optimization by constructing equivalent nodes from feasible divisor sets~\cite{abc}.


\begin{table*}[tb!] 
\caption{Efficiency comparison of different incremental level computation algorithm for \rw on benchmarks}
\label{tab:rewrite_comparison}
\setlength{\tabcolsep}{5pt} 
\centering
\resizebox{.98\textwidth}{!} {
\begin{tabular}{l|rrrr|rrrr|rrrr||c} \toprule
            & \multicolumn{4}{c|}{\inc~\cite{abc}}   & \multicolumn{4}{c|}{\incpq~\cite{KatrielMH05}}   & \multicolumn{4}{c||}{\incour}    \\ 
\cline{2-13}
\multirow{2}{*}[2ex]{Circuit}   &\makecell{\#AND \\ ($10^6$) } &  Level & \makecell{Level m.t. \\ time (s)} & \makecell{All \\ time (s)} & \makecell{\#AND \\ ($10^6$) } &  Level & \makecell{Level m.t. \\ time (s)} & \makecell{ All \\ time (s)} &  \makecell{\#AND \\ ($10^6$) } &  Level & \makecell{Level m.t. \\ time (s)} & \makecell{All \\ time (s)} & \multirow{2}{*}[2ex]{\makecell{Gain node \\ ratio ($\%$)}} \\ \midrule
hyp       & 0.214   & 24,801  & 0.01      & 3.4       & 0.214     & 24,801 & 0.0         & 3.6     & 0.214  & 24801 & 0.0 & 3.4     &0.03   \\
mult256   & 0.551   & 2,293   & 16.7      & 29.9      & 0.551     & 2,293  & 21.5      & 35.5      & 0.551  & 2293  & 0.1 & 13.4     & 7 \\
netcard   & 0.521   & 38      & 8.6       & 24.0      & 0.521     & 38     & 9.0       & 25.3      & 0.521  & 38    & 0.2 & 15.7      & 12 \\
sort1024  & 2.095   & 2,354   & 357.3     & 411.2     & 2.095     & 2,354  & 580.8     & 640.1     & 2.094  & 2150  & 0.4 & 52.1     & 9 \\
mult512   & 2.211   & 4,960   & 55.3      & 114.8     & 2.211     & 4,960  & 50.6      & 112.4     & 2.216  & 4960  & 0.5 & 59.1    & 10 \\
sort2048  & 8.383   & 4,505   & 3,462.1   & 3,717.3   & 8.383     & 4,505  & 6,193.9   & 6,454.2   & 8.382  & 4095  & 1.4 & 215.5   & 11 \\
mult1024  & 8.769   & 9,511   & 1,162.5   & 1,430.5   & 8.769     & 9,511  & 1,621.5   & 1,906.5   & 8.769  & 9511  & 2.3 & 265.1   & 11 \\
sixteen   & 12.178  & 109     & 6,213.5   & 7,290.7   & 12.178    & 109    & 7,205.5   & 8,384.8   & 12.178 & 109   & 3.9 & 1039.4   & 25 \\
twenty    & 15.512  & 111     & 9,372.6   & 10,860.9  & 15.512    & 111    & 11,370.8  & 13,015.2  & 15.512 & 111   & 4.2 & 1425.5  & 25 \\
twentythree & 17.373  & 129     & 11,168.2  & 12,945.9  & 17.373    & 129    & 13,144.4  & 15,132.2  & 17.374 & 129   & 5.0 & 1721.0 & 25  \\
sort4096  & 33.628  & 8,599   & 52,543.5  & 53,674.8  & 33.628    & 8,599  & 92,489.2  & 93,646.7  & 33.545 & 8191  & 5.4 & 948.6   & 11 \\ \midrule
\makecell{Average \\ improvement}   & 1.00    & 1.02    &1,812.6    &10.5      & 1.00      & 1.02   & 2,833.2    &16.5      & 1.00      & 1.00  & 1.00   & 1.00    & 13 \\ \bottomrule  
\end{tabular}
}
\vspace{-4ex}
\end{table*}


\subsection{Efficiency Analysis of Algorithm~\incour}  
To assess the efficiency of our~\incour, we integrated it alongside the baseline \inc~\cite{abc} and \incpq~\cite{KatrielMH05} into two prominent local operators: \rw and \rf.

\begin{figure}[tb!]
    \centering 
    \includegraphics[width=0.7\columnwidth]{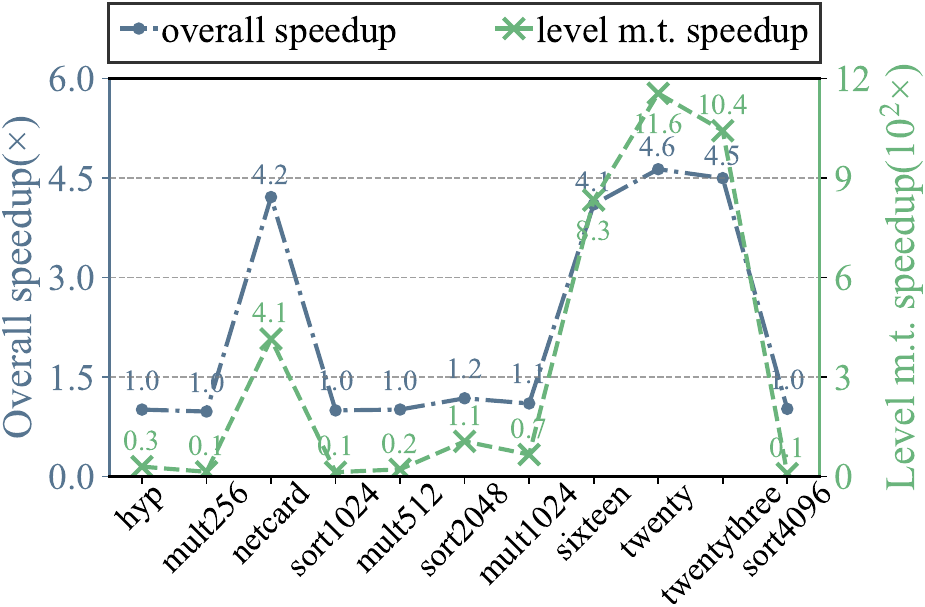} 
    \caption{Efficiency comparison with \inc for \rf}
    \label{fig:exp_efficiency_rf}
\end{figure}
\minisection{Exp-1.1: Performance evaluation with \rw}
To assess the efficiency of our~\incour integrated within \rw, we compare its performance with \inc and \incpq on benchmarks, as shown in~\Cref{tab:rewrite_comparison}. 
Note that, ``Level m.t.~time'' represents the time overhead for maintaining level constraints, while ``All time'' encompasses the total runtime including cut enumeration, resynthesis, graph update, and level maintenance.
From~\Cref{tab:rewrite_comparison}, we have the following findings.

First, when \rw employs \incour for maintaining level constraints, the overall runtime consistently outperforms both \inc and \incpq across circuits of varying scales. 
On average, \incour achieves speedups of 10.5$\times$ and 16.5$\times$ compared to \inc and \incpq, respectively. 
This substantial improvement stems from \incour's  minimal localized traversals, \eg \dynto.  
Specifically, the level maintenance time of \incour is 1,812.6$\times$ and 2,833.2$\times$ faster than \inc and \incpq, respectively. 
These results validate the algorithmic design rationale of \incour as analyzed in~\Cref{sec:analysis}.

Second, the QoR (Quality of Results) achieved by \rw using different level computation algorithms remains comparable across all benchmarks. 
Note that, \incour produces circuits with superior level metrics, achieving an average 2\% reduction compared to \inc and \incpq, primarily contributed by sort4096.
Besides, \incour yields circuits with fewer AND gates, showing an average 0.08\% improvement over the baseline algorithms. 
These minor differences arise from the distinct processing orders of $V$, as \incour follows a partial topological order. 
Moreover, all three algorithms produce nearly identical ratios of rewritable nodes during \rw operations, hence we present a single ``Gain node ratio'' column representing all approaches. 

Third,
\inc outperforms \incpq in most scenarios due to \incpq's overhead from maintaining a priority queue for storing affected nodes requiring updates. \incpq demonstrates superior performance over \inc only when the AIG exhibits both a very large maximum level $\mathcal{L}_{\max}$ and a relatively small gain node ratio, \eg hyp and mult512, as analyzed in~\Cref{sec:analysis}.

\minisection{Exp-1.2: Performance evaluation with refactor}
To further evaluate the efficiency of  \incour with \rf, we compare its performance with \inc and \incpq on benchmarks, with the comparative results illustrated in~\Cref{fig:exp_efficiency_rf}.
Note that, the final AIG size and depth obtained from \rf when using \inc, \incpq, and \incour are comparable.
Due to space limitations, \Cref{fig:exp_efficiency_rf} only illustrates the speedup in overall time and level maintenance time achieved by \incour compared to \inc.
From~\Cref{fig:exp_efficiency_rf}, we have the following findings.

\begin{table}[tb!]
\caption{Efficiency comparison with \incour for \rss}
\label{tab:efficiency_resub}
\setlength{\tabcolsep}{2pt} 
\resizebox{.98\columnwidth}{!} {
    \begin{tabular}{l|rr|rr|rr} \toprule
         & \multicolumn{2}{c|}{\inc} & \multicolumn{2}{c|}{\incpq}  & \multicolumn{2}{c}{\incour} \\  \cline{2-7}
    \multirow{2}{*}[2ex]{Circuit}  & \makecell{Level m.t. \\ time (s)} & \makecell{All \\ time (s)} & \makecell{Level m.t. \\ time (s)} & \makecell{All \\ time (s)} & \makecell{Level m.t. \\ time (s)} & \makecell{All \\ time (s)}  \\ \midrule
    hyp                       & 0.9            & 2.3      & 0.1            & 1.6      & 0.0             & 1.5       \\
    mult256                   & 3.6            & 9.7      & 3.6            & 10.1     & 2.8             & 8.9       \\
    netcard                   & 0.0            & 50.4     & 0.0            & 52.8     & 0.0             & 52.2      \\
    sort1024                  & 47.0           & 74.1     & 97.2           & 126.0    & 40.1            & 67.8      \\
    mult512                   & 36.9           & 62.2     & 37.3           & 64.3     & 28.5            & 54.6      \\
    sort2048                  & 1,021.3        & 1,127.9  & 1,656.1        & 1,772.4  & 463.8           & 574.6     \\
    mult1024                  & 601.5          & 706.1    & 692.8          & 804.0    & 561.1           & 665.5     \\
    sixteen                   & 0.2            & 207.5    & 0.2            & 222.7    & 0.6             & 207.5     \\
    twenty                    & 0.4            & 279.3    & 0.3            & 313.7    & 0.8             & 284.7     \\
    twentythree               & 0.3            & 322.8    & 0.2            & 338.7    & 0.9             & 328.8     \\
    sort4096                  & 29,742.6       & 30,219.1 & 38,366.3       & 38,833.9 & 11,304.9        & 11,764.3  \\ \midrule
    \makecell{Avg. impro.}                   & 2.5            & 2.4      & 3.3            & 3.0      & 1               & 1        \\  \bottomrule
    \end{tabular}
}
\end{table}

First, when \rf employs \incour for maintaining level constraints, \incour consistently outperforms both baselines. Compared to \inc, \incour achieves an average speedup of 2.3$\times$ in overall time, with a corresponding 337.0$\times$ improvement in level update time. 
Similarly, \incour demonstrates superior performance over \incpq, delivering an average speedup of 2.4$\times$ in overall time and 365.7$\times$ in level maintenance time.
  
Second,  the optimization potential of \rf is generally more limited compared to \rw, as evidenced by the average gain node ratio of only 4\% for \rf versus 13\% for \rw (as detailed in Exp-1.1). 
When the overall time improvements are marginal, the corresponding gain node ratios tend to be very small, \eg hyp, sort1024, and sort4096 exhibit gain node ratios of merely 0.9\%, 0.4\%, and 0.2\%, respectively.  
Therefore, \incour achieves greater efficiency improvements on operations with higher optimization potential, which naturally require longer baseline runtimes.

\subsection{Scalability Analysis of Algorithm~\incour}  
\minisection{Exp-2: Performance Evaluation with~\rss}
To evaluate the scalability of our level maintenance algorithm, we extend it to~\rss. 
For \rss, we adapt \incour to use a priority queue for forward level propagation due to the transitive fanout cone requirements, while maintaining the original approach for order maintenance~\dynto and reverse level computation~\dynrl.
The comparative results with \inc and \incpq are presented in~\Cref{tab:efficiency_resub}, yielding the following findings.

First, 
the adapted \incour consistently outperforms both \inc and \incpq across nearly all benchmarks. On average, \incour achieves speedups of 2.4$\times$ and 3.1$\times$ in overall runtime compared to \inc and \incpq, respectively. The corresponding level maintenance time shows efficiency improvements of 2.5$\times$ and 3.3$\times$, respectively.

Second, 
the efficiency gains of the adapted \incour in \rss are less pronounced compared to those observed in \rw~\Cref{tab:rewrite_comparison}. This reduction stems from the modified forward level update, which involves priority queue operations and cannot be bounded by $O(\Delta \log \Delta )$ time, due to the potential for wider-reaching fanout dependencies.

\begin{table*}[tb!]
\setlength{\tabcolsep}{4pt} 
\caption{Affected region comparison of \incour for \rw and \rf}
\label{tab:traversal_number}
\centering
\resizebox{.98\textwidth}{!} {
\begin{tabular}{l|rrr|rrrr||rrr|rrrr} \toprule
 & \multicolumn{3}{c|}{\inc for \rw} & \multicolumn{4}{c||}{\incour for \rw} & \multicolumn{3}{c|}{\inc for \rf} & \multicolumn{4}{c}{\incour for \rf} \\
 \cline{2-15}
\multirow{2}{*}[2ex]{Circuit}   & \#\kw{NTL}    & \#\kw{NTR}    & \makecell{Level m.t. \\ time (s)}      &\#\kw{NTL}     & \#\kw{NTR}   & \#\kw{IPTO}  &\makecell{Level m.t. \\ time (s)}  & \#\kw{NTL}    & \#\kw{NTR}    & \makecell{Level m.t. \\ time (s)}   & \#\kw{NTL}     & \#\kw{NTR}   & \#\kw{IPTO}  &\makecell{Level m.t. \\ time (s)}    \\ \midrule
hyp                       & 0.001    & 0.005    & 0.01     & 2.0    & 0.001 & 0.001   & 0.01 & 0.2     & 0.2       & 0.3       & 2.0   & 0.04 & 0.02   & 0.01 \\
mult256                   & 91.7     & 165.3    & 16.7      & 2.2    & 0.4   & 0.3     & 0.1  & 1.6     & 0.9       & 1.0       & 2.1   & 0.2   & 0.1     & 0.1   \\
netcard                   & 0.7      & 470.1    & 8.6       & 2.0    & 0.4   & 0.009   & 0.2  & 0.013   & 2,799.9   & 45.6      & 2.0   & 0.5   & 0.011   & 0.1   \\
sort1024                  & 480.6    & 934.3    & 357.3     & 2.1    & 0.3   & 0.2     & 0.4  & 0.2     & 4.8       & 1.3       & 2.0   & 0.004 & 0.002   & 0.1   \\
mult512                   & 102.2    & 42.3     & 55.3      & 2.2    & 0.4   & 0.3     & 0.5  & 0.1     & 0.6       & 5.3       & 2.1   & 0.1   & 0.1     & 0.2   \\
sort2048                  & 1,435.9  & 1,698.2  & 3,462.1   & 2.1    & 0.3   & 0.2     & 1.4  & 50.4    & 0.03       & 50.9      & 2.0   & 0.008 & 0.008   & 0.5   \\
mult1024                  & 415.0    & 465.3    & 1,162.5   & 2.2    & 0.4   & 0.4     & 2.3  & 1.7     & 1.0       & 76.5      & 2.1   & 0.2   & 0.1     & 1.2   \\
sixteen                   & 4.8      & 3,528.3  & 6,213.5   & 2.5    & 0.5   & 0.3     & 3.9  & 1.0     & 1,208.4   & 2,391.8   & 2.5   & 0.3   & 0.2     & 2.9   \\
twenty                    & 5.9      & 3,923.6  & 9,372.6   & 2.5    & 0.5   & 0.3     & 4.2  & 1.3     & 1,417.0   & 3,777.5   & 2.5   & 0.3   & 0.3     & 3.3   \\
twentythree               & 6.4      & 4,133.3  & 11,168.2  & 2.5    & 0.5   & 0.3     & 5.0  & 1.5     & 1,529.9   & 4,482.6   & 2.5   & 0.3   & 0.3     & 4.3   \\
sort4096                  & 5,695.0  & 7,960.6  & 52,543.5  & 2.1    & 0.3   & 0.2     & 5.4  & 0.1     & 0.2       & 13.2      & 2.0   & 0.004 & 0.002   & 1.5   \\ \midrule
\makecell{Average}    & 748.9    & 2,120.1  & 7,669.1   & 2.2    & 0.4   & 0.2     & 2.1  & 5.3     & 633.0     & 986.0     & 2.2   & 0.2   & 0.1     & 1.3    \\ \bottomrule
\end{tabular} 
}\\
\#\kw{NTL} and \#\kw{NTR} refer to the average number of \underline{n}odes \underline{t}raversed during \underline{l}evel computation and \underline{r}everse level computation, respectively, (\ie normalized by~$|V|$). 
\#\kw{IPTO} denotes the average number of nodes requiring \underline{i}nvalid \underline{p}artial \underline{t}opological \underline{o}rder updates within~\incour.
\vspace{-4ex}
\end{table*}
 
\begin{figure}[tb!]
    \centering
    \subfloat[\rfz\label{subfig:rfz}]{
		\includegraphics[width=0.486\columnwidth]{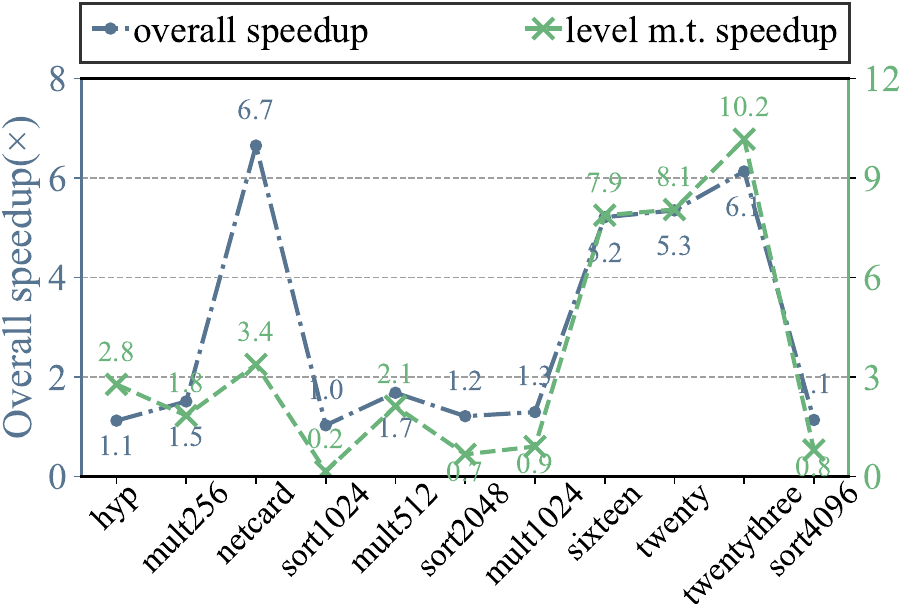}} 
    \subfloat[\rwz\label{subfig:rwz}]{
		\includegraphics[width=0.498\columnwidth]{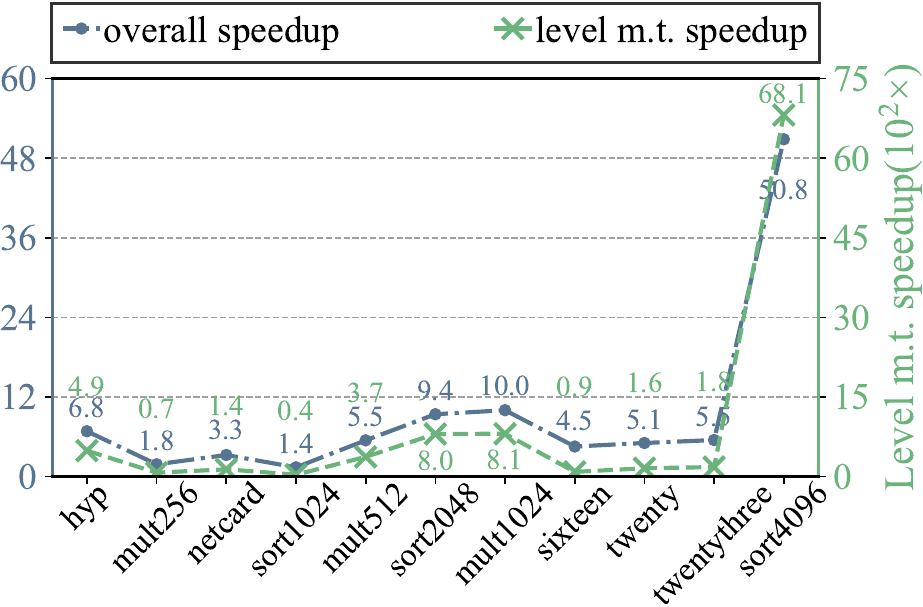}} 
    \caption{Efficiency comparison with \inc using zero gain.}
    \label{fig:eff_zero_gain}
\end{figure}

\subsection{Boundedness Analysis of Algorithm~\incour}  
\minisection{Exp-3: Affected region comparison of~\incour} 
To validate the boundedness of~\incour, we analyze the average affected regions during level maintenance, as shown in~\Cref{tab:traversal_number}. 
It presents the per-node statistics of traversed regions for both \inc and \incour across \rw and \rf operations, revealing the following key findings.

First, for both \rw and \rf operations, \incour consistently maintains small constant values for the average affected regions (\#\kw{NTL}, \#\kw{NTR}, and \#\kw{IPTO}), regardless of AIG size. 
This aligns with our theoretical analysis in~\Cref{sec:method_single}, where \#\kw{NTL} typically remains around 2-3 nodes, consistent with the constant in~\Cref{the:level_complexity}. 
Furthermore, the level m.t.~time of \incour scales linearly with AIG size in \rw and \rf, \eg  as the AIG size increases from 0.214$\times 10^6$ to 41.9$\times 10^6$ nodes, the level m.t. time grows from 0.01s to 5.4s. 
This linear relationship is consistent with $O(|V| \Delta \log \Delta )$ established in~\Cref{the:main_complexity}, as $\Delta $  is very small in practice,
validating the analysis in~Theorems~\ref{the:topo_order_complexity}, \ref{the:level_complexity}, and \ref{the:reverse_complexity}.

Second, in contrast, \inc exhibits unbounded behavior for both level and reverse level computations, where even the affected regions cannot be constrained by the modified subgraph or its neighborhood. 
As the circuit scale increases, the affected regions grow substantially, thereby \inc is an incrementally unbounded algorithm~\cite{Ramalingam1996dynamic, fan2022inc}. 
This fundamental difference results in \inc requiring an average of 7,669.1s for level constraint maintenance in \rw operations, while \incour maintains the level in merely 2.1s.

\subsection{Parameters Analysis of Algorithm~\incour}   
\minisection{Exp-4: Zero-gain parameter evaluation of~\incour}   
To further assess the robustness of \incour under varying optimization scenarios, we evaluate the efficiency of \rw and \rf with \emph{-z} parameter enabled, as shown in~\Cref{fig:eff_zero_gain}.
It allows subgraph transformations even when zero node gain is achieved, thereby expanding the structural exploration space of AIG optimization under aggressive strategies. 
From~\Cref{fig:eff_zero_gain}, we have the following findings.

First,  \Cref{subfig:rfz} shows that \incour maintains its performance advantage for \rfz.
Compared to \inc, it delivers an average overall runtime speedup of 2.9$\times$ and reduces level maintenance time by a factor of 352.3$\times$. 

Second, as shown in~\Cref{subfig:rwz},  \incour substantially outperforms \inc during \rwz operation.
It achieves an average overall runtime speedup of 9.5$\times$, accompanied by a remarkable 904.7$\times$ reduction in level maintenance time.
These gains mirror those seen with standard \rw (Exp-1.1), underscoring the consistent efficiency of \incour even when exploring a broader optimization search space.

\subsection{Runtime Profile Analysis of Algorithm~\incour} 
\minisection{Exp-5: Comparative runtime breakdown of~\incour} 
To analyze the performance impact of our~\incour, we compare runtime breakdowns against \inc, as shown in~\Cref{fig:runtime_breakdown}.
It shows the average time distribution across cut enumeration, resynthesis, level maintenance, and graph update stages for \rw, \rf, and \rss on the benchmark.
From~\Cref{fig:runtime_breakdown}, we have the following findings.

\incour dramatically reduces level maintenance overhead for \rw and \rf operations. For \rw, level maintenance time drops from 7,669s (93.2\% of total) to 2.1s (0.4\% of total). Similarly, for \rf, it decreases from 985s (63.7\% of total) to 1.3s (0.2\% of total).
Although \rss cannot achieve the same $O(|V| \Delta \log \Delta )$, \incour still delivers significant gains, reducing level maintenance time from 2,859s to 1,127.6s (2.5$\times$ speedup).
With \incour, level maintenance becomes negligible for \rw and \rf, thereby establishing a foundation for tractable optimization of large-scale circuits.

\begin{figure}[tb!] 
    \centering
    \includegraphics[width=.7\linewidth]{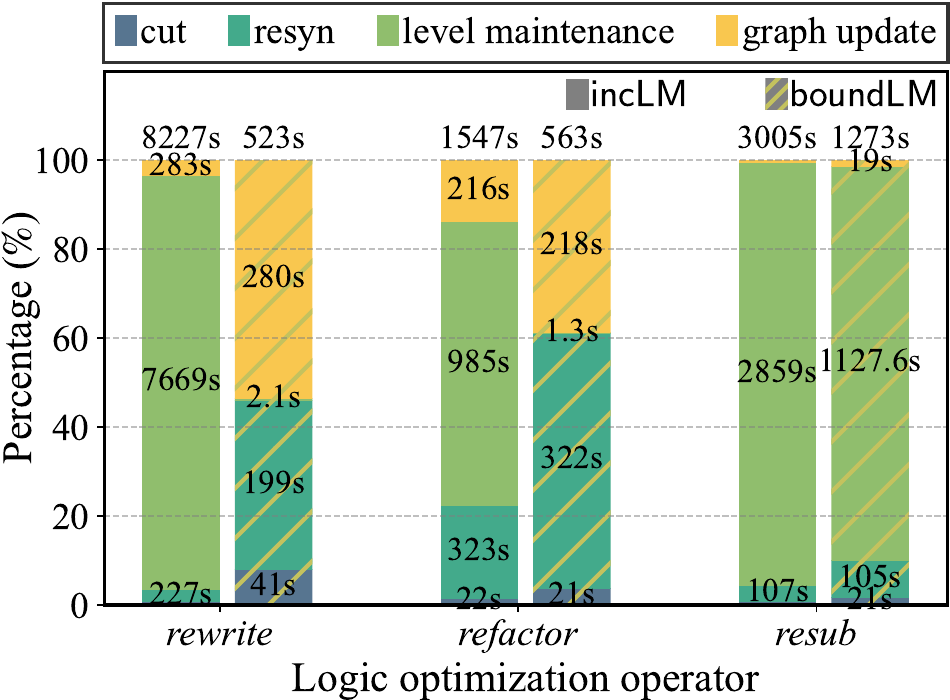}
    \caption{Proportion of time for each part in logic optimization.} 
    \label{fig:runtime_breakdown}   
\end{figure}

\section{Conclusion}
\label{sec:conclusion} 
We analyzed the dynamic level maintenance problem in iterative, level-constrained logic optimization by reframing it through partial topological order maintenance. 
Leveraging this insight, we developed a bounded algorithm, \incour, for dynamically maintaining topological order (\dynto), node levels (\dynlev), and reverse levels (\dynrl), in $O(|V| \Delta \log\Delta)$ time for $|V|$ updates.   
As a result, on large-scale benchmarks, \incour achieves an average 6.4$\times$ overall speedup in \rw and \rf, driven by a 1074.8$\times$ acceleration in the level maintenance, without any degradation in QoR.

This work establishes a new paradigm for enhancing logic synthesis by applying principles from dynamic graph algorithms. 
It provides a foundation for developing more scalable and efficient optimization tools capable of handling the complexity of next-generation integrated circuits.


{
  \bibliographystyle{abbrv}
  \bibliography{ref/all}
}

\vspace{-8ex}
\begin{IEEEbiography}[{\includegraphics[width=1in,height=1.25in,clip,keepaspectratio]{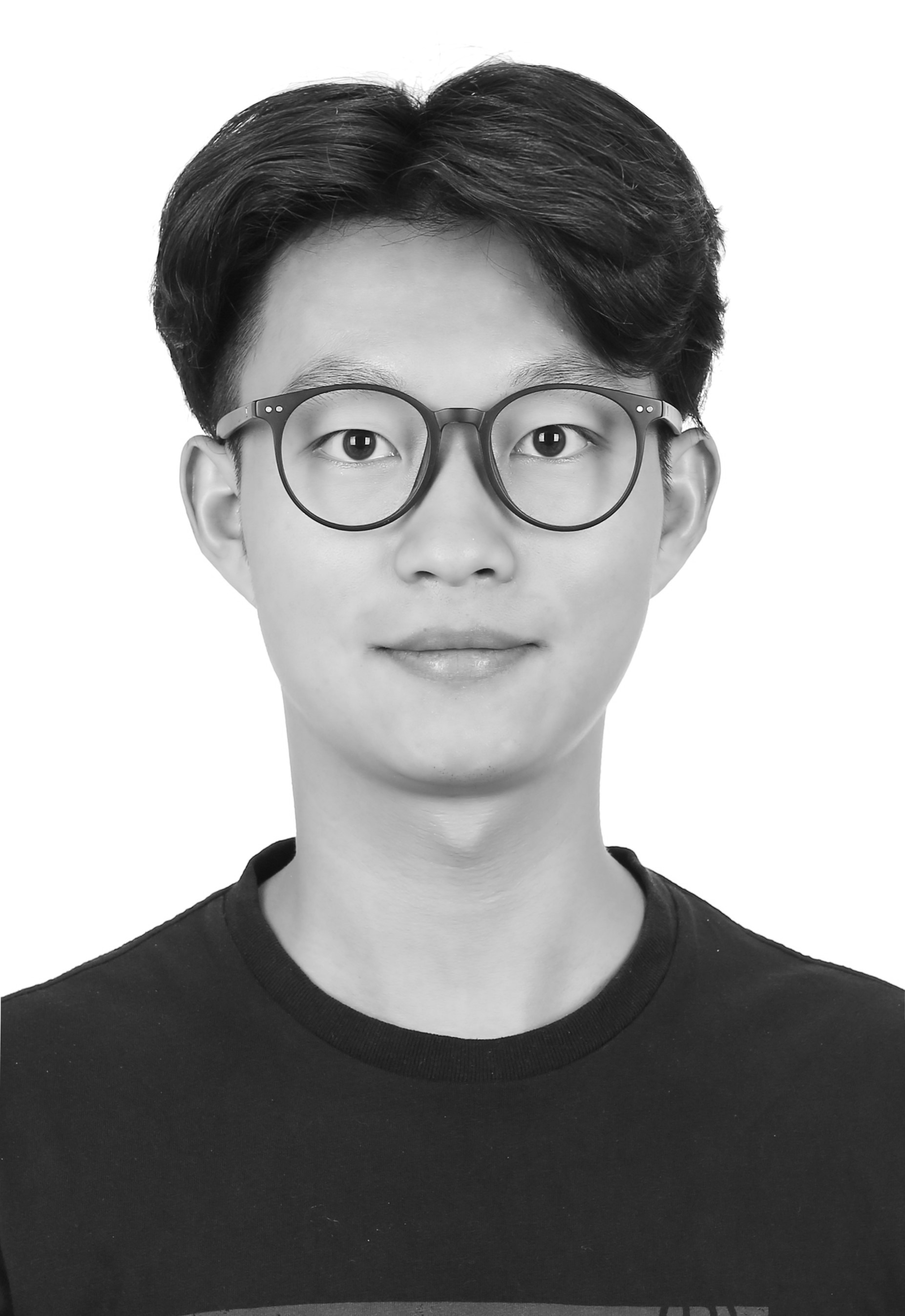}}] 
{Junfeng Liu} received the PhD degree in computer science from Beihang University, China, in 2024. 
He is currently a postdoctoral researcher at Pengcheng Laboratory, China.
His research interests include EDA logic synthesis and graph data management. He has published over 10 papers in journals and conferences such as TCAD, TKDE, TODAES, ICCD, WSDM, and CIKM.
\end{IEEEbiography}

\begin{IEEEbiography}[{\includegraphics[width=1in,height=1.25in,clip,keepaspectratio]{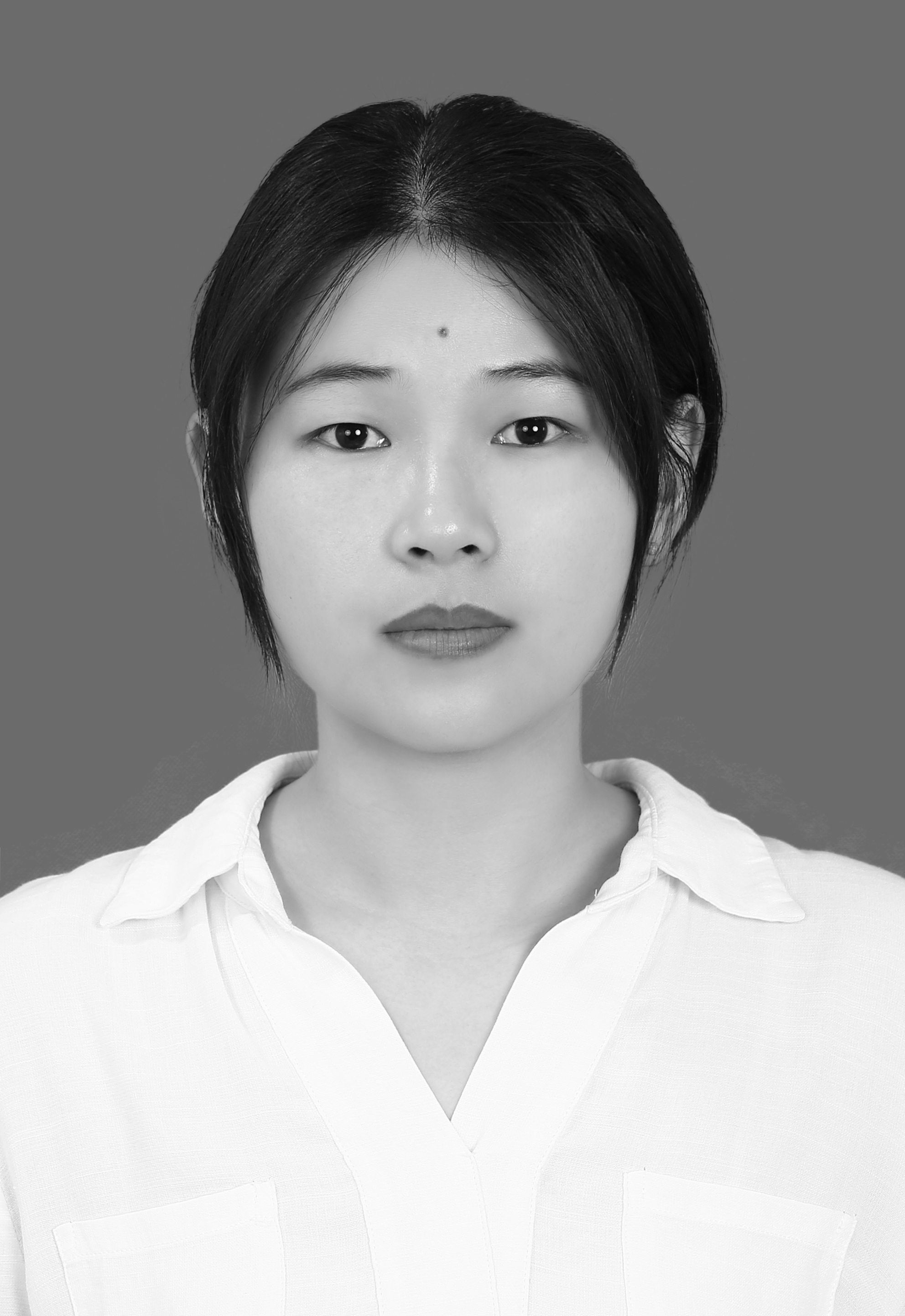}}] 
{Qinghua Zhao} received the PhD degree in computer science from Beihang University, China, in 2024. 
She previously worked as a visiting scholar in the Coastal NLP Group at the Department of Computer Science, University of Copenhagen, Denmark, in 2023.
She is currently a lecturer at the School of Artificial Intelligence and Big Data, Hefei University, Hefei, China.
Her research interests include data-driven AI, NLP, and computer-aided design.
\end{IEEEbiography}

\vspace{-8ex}
\begin{IEEEbiography}[{\includegraphics[width=1in,height=1.25in,clip,keepaspectratio]{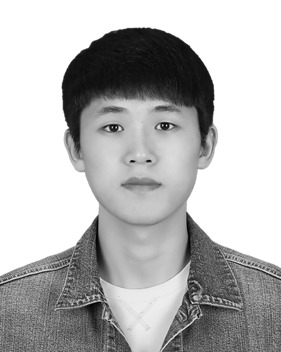}}] 
{Liwei Ni} received the B.S. degree in computer science from the Anhui University of Finance and Economics, Bengbu, China, in 2018, and the M.S. degree in Software Engineering from Beihang University, Beijing, China, in 2021. 
He is pursuing the Ph.D. degree with the Institute of Computing Technology, Chinese Academy of Sciences, Beijing, China, and is jointly trained with Pengcheng Laboratory.
His research focuses on logic synthesis.
\end{IEEEbiography}

\vspace{-8ex}
\begin{IEEEbiography}[{\includegraphics[width=1in,height=1.25in,clip,keepaspectratio]{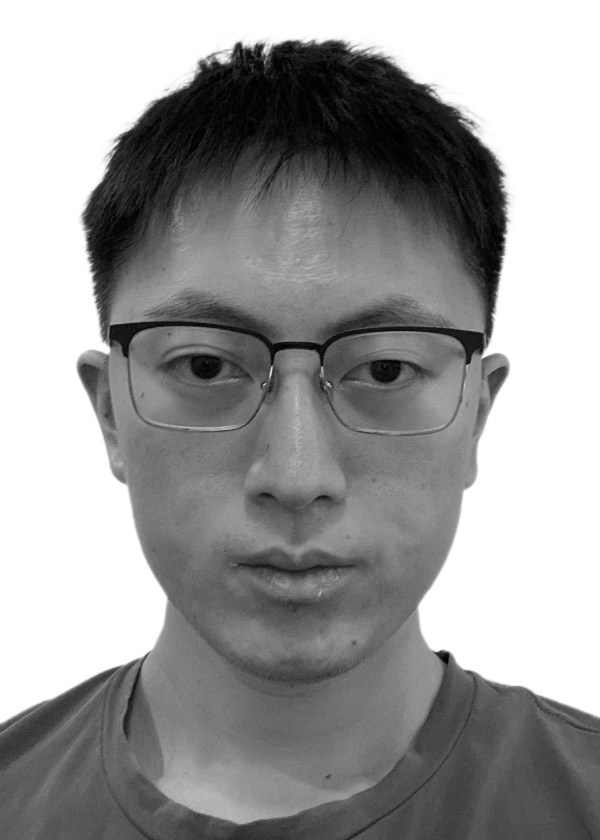}}] 
{Jingren wang}
received his Master's degree in computing science from the University of Glasgow, Scotland, UK, in 2022. He is currently a Research Assistant at the Hong Kong University of Science and Technology (Guangzhou), China. His research focuses on logic synthesis, with particular interests in Boolean algebra and combinational optimization.
\end{IEEEbiography}

\vspace{-8ex}
\begin{IEEEbiography}[{\includegraphics[width=1in,height=1.25in,clip,keepaspectratio]{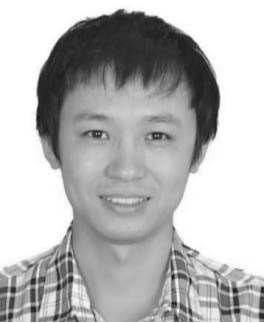}}]
{Biwei Xie} received his Ph.D. degree from the Institute of Computing Technology, Chinese Academy of Sciences, in 2018.
He is an Associate Professor at the same institution. 
His research interests encompass open EDA, open-source chip design, high-performance computing, and computer architecture. 
His work has been published in leading international conferences such as CGO, ICS, ICCAD, and DATE.
\end{IEEEbiography}
 
\vspace{-8ex}
\begin{IEEEbiography}[{\includegraphics[width=1in,height=1.25in,clip,keepaspectratio]{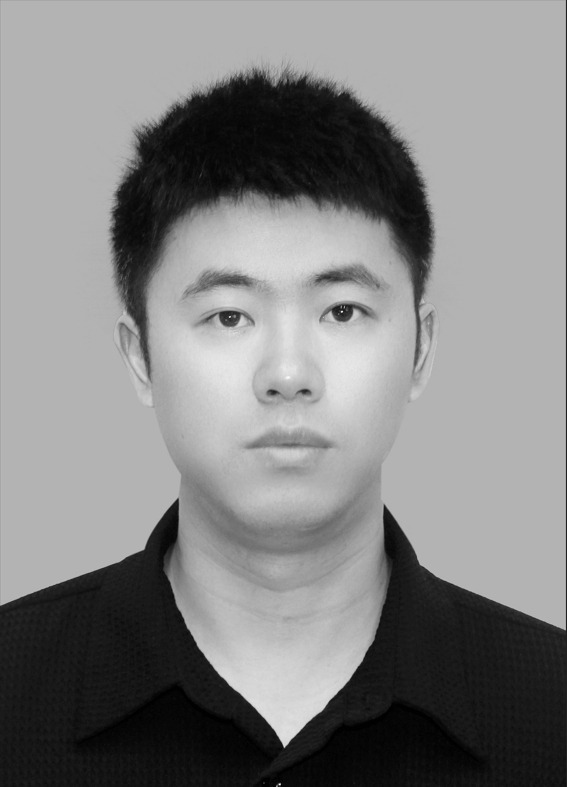}}] 
{Xingquan Li}
received the Ph.D degree from Fuzhou University, China in 2018. He is currently an Associate Researcher at Pengcheng Laboratory. His research interests include EDA and AI for EDA. His team has developed an open-source infrastructure of EDA and toolchain (iEDA). He has published over 60 papers in journals and conferences such as TCAD, TC, TVLSI, TODAES, DAC, ICCAD, DATE, ICCD, ASP-DAC, ISPD, NIPS, etc.  He received the Best Paper Award from ISEDA 2023.
\end{IEEEbiography}

\vspace{-8ex}
\begin{IEEEbiography}  [{\includegraphics[width=0.8in,height=1in,clip,keepaspectratio]{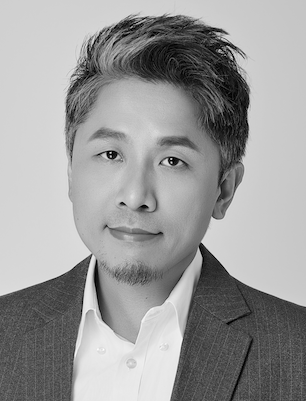}}]
    {Bei Yu}
    (M'15-SM'22)
    received the Ph.D.~degree from The University of Texas at Austin in 2014.
    He is currently an Associate Professor in the Department of Computer Science and Engineering, The Chinese University of Hong Kong.
    He has served as TPC Chair of ACM/IEEE Workshop on Machine Learning for CAD, and in many journal editorial boards and conference committees.
    He received eleven Best Paper Awards from ICCAD 2024 \& 2021 \& 2013,
    IEEE TSM 2022, DATE 2022, ASPDAC 2021 \& 2012, ICTAI 2019, Integration, the VLSI Journal in 2018,
    ISPD 2017, SPIE Advanced Lithography Conference 2016, six ICCAD/ISPD contest awards,
    and many other awards, including DAC Under-40 Innovator Award (2024), IEEE CEDA Ernest S.~Kuh Early Career Award (2022), and Hong Kong RGC Research Fellowship Scheme (RFS) Award (2024).
\end{IEEEbiography}

\vspace{-8ex}
\begin{IEEEbiography}[{\includegraphics[width=1in,height=1.25in,clip,keepaspectratio]{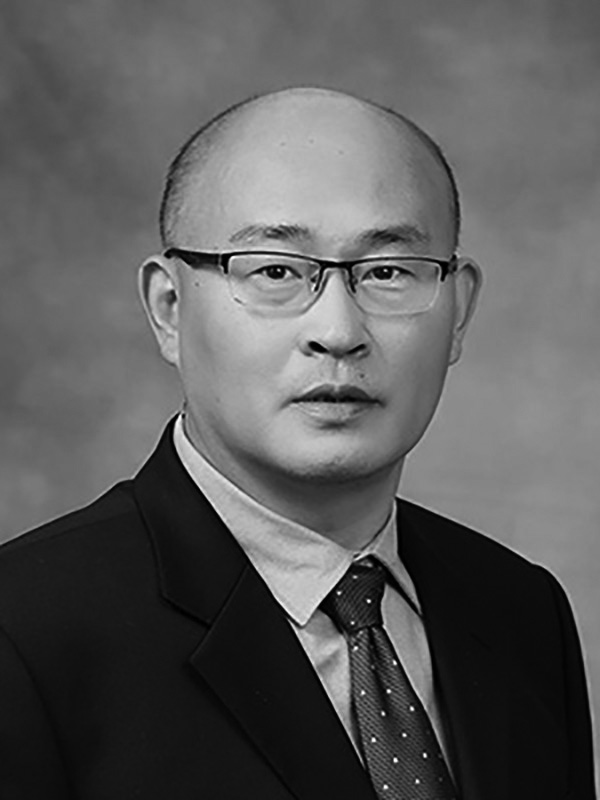}}]  
{Shuai Ma} (Senior Member, IEEE) received the Ph.D. degrees in computer science from Peking University, China, in 2004, and from The University of Edinburgh, England, in 2010, respectively. He is a professor with the School of Computer Science and Engineering, Beihang University, China. He was a postdoctoral research fellow with the Database Group, University of Edinburgh, a summer intern at Bell Labs, Murray Hill, NJ, and a visiting researcher of MSRA. His current research interests include big data, database theory and systems, data cleaning and data quality, and graph data analysis.
\end{IEEEbiography}

\end{document}